\documentclass[notitlepage, 10point, final]{article}
\usepackage[dvipsnames]{xcolor}
\usepackage{lscape}
\usepackage{soul}
\usepackage{amsmath,amssymb}
\usepackage{bbm}
\usepackage{amsthm}
\usepackage{graphicx}
\usepackage{color}
\usepackage[left=1in, right=1in, top=1in, bottom=1in]{geometry}
\usepackage{authblk}
\usepackage[square,sort,comma,numbers]{natbib}
\usepackage{tikz}
\usepackage[normalem]{ulem}
\usepackage{cancel}

\definecolor{airforceblue}{rgb}{0.36, 0.54, 0.66}
\definecolor{antiquefuchsia}{rgb}{0.57, 0.36, 0.51}
\definecolor{ballblue}{rgb}{0.13, 0.67, 0.8}
\definecolor{brown(web)}{rgb}{0.65, 0.16, 0.16}
\definecolor{brown(traditional)}{rgb}{0.59, 0.29, 0.0}
\newcommand{\jl}[1]{\textcolor{black}{#1}}

\newcommand{\jlc}[1]{\textcolor{black}{#1}}
\newcommand{\ay}[1]{\textcolor{black}{#1}}
\newcommand{\seth}[1]{\textcolor{black}{#1}}			
\newcommand{\jld}[1]{\textcolor{black}{#1}}
\newcommand{\ayd}[1]{\textcolor{black}{#1}}
\newtheorem{theorem}{Theorem}
\newtheorem*{theorem1}{Silverman-Brown Limit Law \cite{lao}}
\newtheorem{proposition}{Proposition}
\title{On collisions times of `\jld{self-}sorting' interacting particles in one-dimension with random initial positions and velocities}
\author[1,2]{Joceline Lega}
\author[1]{Sunder Sethuraman}
\author[1,2]{Alexander L Young}
\affil[1]{Department of Mathematics, University of Arizona}
\affil[2]{Program in Applied Mathematics, University of Arizona}
\date{\today}
\begin{document}
\maketitle
\thispagestyle{empty}
\begin{abstract}
We investigate a one-dimensional system of $N$ particles, initially distributed with random positions and velocities, interacting through binary collisions.  The collision rule is such that there is a time after which the $N$ particles do not interact and \jld{become} sorted according to their velocities.
When the collisions are elastic, we derive asymptotic distributions for the final collision time \jld{of} a single particle and the final collision time of the system as the number of particles approaches infinity, under different assumptions \jld{for} the initial distributions of the particles' positions and velocities.  \ay{\jld{For comparison}, a numerical investigation is carried out to determine how \jld{a non-}elastic collision rule, which conserves neither momentum nor energy, affects 
the \jld{median} collision time of a particle and the \jld{median} final collision time of the system.}
\end{abstract}
\section{Introduction}

\seth{\ay{\jld{We c}\seth{onsider} a collection of $N$ `identical' point-particles, with equal mass, moving \seth{on $\mathbb{R}$} and interacting through \jld{a linear} binary collision \ayd{rule} \jld{given by}
\begin{equation}
\begin{split}
v_i' &= (1-\epsilon)v_j \\ 
v_j' &= (1-\epsilon)v_i . \\ 
\end{split}
\label{eq:coll_rule}
\end{equation}
Here $v_i,v_j$ ($v_i',v_j'$) are the pre (post) collision velocities of particles $i$ and $j$ \jld{and the} parameter  $\epsilon < 1$ controls conservation and dissipation of momentum and energy. Between collisions, particles undergo free flight.}} \jld{When $\epsilon=0$, \ayd{a collision is elastic, preserving momentum and energy, and corresponds to an exchange of velocity between the two colliding particles.}} \ayd{For $\epsilon \ne 0$, we refer to the collisions as non-elastic since Eq. (\ref{eq:coll_rule}) differs from the traditional construction of inelastic collisions which conserve momentum but not energy. In Section \ref{sec:inelastic_intro}, we discuss the motivations for choosing Eq. (\ref{eq:coll_rule}) and its connection to standard inelastic collisions through a generalized linear collision framework. }When
\seth{$\epsilon <0$, \jld{collisions} generate energy, and when $\epsilon \in (0,1)$ \jld{they} dissipate energy. The case $\epsilon =1$ is degenerate, in the sense that particles stop their motion and remain `\jld{frozen}' together. 
 For any $\epsilon < 1$, as we show in Sections \ref{sec:sorting} and \ref{sec:inelastic_intro}, each particle experiences a final collision, as eventually the velocities of \jld{all of} the particles \jld{become} sorted. In the above context, it is natural to ask (1) when a particle will experience its final collision, (2) when the final collision of the entire collection will occur, and (3) how these statistics depend on the initial position and velocity distributions, the number of particles in the system, and $\epsilon.$  \jld{The purpose of this article is to address these questions, analytically in the case of elastic collisions ($\epsilon = 0$), and numerically for non-elastic collisions ($\epsilon \ne 0$).}}

\jld{The motivations for this work are twofold. First, }\seth{in a recent \jld{article} by Bardos et al \cite{levermore}, the long time limit of solutions to the Boltzmann equation over $\mathbb{R}^d$ was investigated.  In the case of particles interacting through \ay{hard} collisions, it was found that in unbounded domains, as opposed to bounded domains where the phenomenon is different, the dispersive effects of particle free flight\jld{s} were sufficient to quench the dissipative effects of collisions\jld{, thereby} preventing the system from reaching \ay{a state of maximal entropy}.  \jld{Specifically, i}n terms of microscopic dynamics, one can envision a finite collection of particles in $\mathbb{R}^d$ with random initial positions and velocities.  As time grows, the particles will likely spread out and interact with each other until a time when no more collisions occur.  At this point, the system would be in \jld{a}  steady-state different from \seth{what} might arise if the particles were kept confined to a bounded domain in $\mathbb{R}^d$.  Therefore, from this view, understanding the final time of collision in the system is a natural question.}
\seth{\jlc{
\jld{Second,} microscopic descriptions in terms of non-overlapping particles have proven to be central to the modeling of diffusion in single-file systems (see e.g. \cite{Roedenbeck98} and references therein), single-lane traffic flow \cite{Helbing96}, and self-organization in one-dimensional systems of self-propelled particles \cite{Czirok99}. \jld{\seth{In this context,} the system \seth{we study} can be viewed as the microscopic equivalent of a one-dimensional gas of point particles, which interact through a generalized collision rule that is not necessarily subject to conservation of momentum or energy.} A related example of dynamics on one-dimensional lattices includes the Stirring exclusion process \cite{Ligget1985}. 
}}

\seth{Understanding the final collision time of a distinguished particle and the final collision time in \jld{a} system of $N$ particles are questions about maximal order statistics for certain functions of the initial positions and velocities.  Although we will assume the initial position and velocities of the particles are independent and identically distributed, when collisions are elastic, the collision times turn out to be an array of $\binom{N}{2}$ non-independent `exchangeable' random variables without a finite mean. In this setting, we are able to analyze the collision times rigorously.  When the collisions are \jld{non-}elastic, however, an analytic approach is more difficult. As such, we use molecular dynamics simulations to \jld{assess} the effects of \jld{non-}elastic collisions on the collision times.  In both \jld{non-}elastic and elastic collisions, to avoid degeneracies, we will assume that the distributions of the initial positions and velocities are continuous random variables.}

\seth{Informally, in \jld{the case of} elastic collisions one set of our main results is that the final collision time of a distinguished particle, in the system of $N$ particles, depends on the moment properties of the position random variable (Theorems \ref{FTC_particle1}, \ref{FTC_particle2}, \ref{FTC_particle3}).  That is, this final time scales with $N$ when the position r.v. has finite mean. It scales with $N^{1/\alpha}$ when the position r.v \jld{has} the form of a symmetric stable law with parameter $0<\alpha<1$, and it scales with $N\log(N)$ when the position r.v. is a Cauchy distribution. In \jld{both of} these cases, the limit distribution is a mixture of Fr\'echet distributions.}
\seth{Another set of our results concerns the final collision time \jld{$T^{(N)}$ for} the whole system of $N$ particles.  \jld{Here a}gain, this time depends on how many moments the position r.v. possesses.  In particular, when the \jld{latter} has at least a $3/2$-moment, \jld{$T^{(N)}$} scales with $\binom{N}{2}$. When the position r.v. \jld{has} the form of a stable symmetric law with parameter $0<\alpha<1$, \jld{$T^{(N)}$} scales \jld{with} $N^{2/\alpha}$, and when the position r.v. is a Cauchy distribution, \jld{$T^{(N)}$} scales with $N^2\log N$.  In each of these cases, the limit is of Fr\'echet type (Theorems \ref{FTC_system}, \ref{FTC_system2}, \ref{FTC_system3}).  However, we also show that when the position r.v. has a first moment, the sequence of final collision times for $N\geq 2$ is tight in the scale $\binom{N}{2}$ (Proposition \ref{prop:tight}). \jld{Moreover, s}ome numerical simulations are provided which indicate that the result of Theorem \ref{FTC_system} should extend to the case when the position r.v. has a mean without further restrictions.}

\seth{\jld{For non-elastic collisions, we present numerical results in the case where} the initial positions are standard normal distributions. Numerical simulations indicate that, unlike elastic collisions, the final collision \jld{time} of \jld{a distinguished} particle does not scale with $N$, \jld{and} the final collision time of the system \jld{does not} scale \jld{with} $N^2$. Instead, both times require an exponential correction depending on $N\epsilon$ ad the initial distribution of velocities.  An ansatz is discussed which offers a limited explanation to the observed changes.}

\seth{A\jld{n outcome} of this work \jld{is the formulation of} a `sorting model' through which the collision times under elastic interactions can be conveniently analyzed.  In this \ay{elastic interaction} framework, the $\binom{N}{2}$ collision times \jld{are understood} through methods for exchangeable arrays.  \ayd{Various types of these arrays have been considered in \cite{berman, Barbour_Eagleson,lao,silverman}.}  \ayd{From the perspective of colliding particles, this work is related to a} different collection of processes, where the initial positions are nonrandom, but the particles interact randomly, \ayd{which} \jld{was} considered in \cite{angel, Angel2012,Angel2007,Angel2009}.}

\jld{The rest of this article is organized as follows.} \seth{We now state precisely, in the next subsection, the setting and the main quantities of interest, recast the space-time dynamics of the $N$ particles in terms of a `sorting process', and formulate the questions studied \ayd{with respect to elastic collisions}.   In subsection \ref{sec:inelastic_intro}, we prove the existence of a final collision time under \jld{non-}elastic collisions ($\epsilon \ne 0$).  In Section \ref{sec:results}, we state the main theorems for elastic collisions.  In Section \ref{sec:sim}, we discuss numerical results for \jld{non-}elastic collisions.  Finally, the proofs of theorem 1, 2, 4, and 5 are given in Section \ref{sec:proof}.  As the proofs of theorems 3 and 6 are similar in structure to those of theorems 2 and 5, we have elected to include them in an appendix.}

\subsection{Elastic collisions on the line as a sorting process}
\label{sec:sorting}
\jl{Consider $N$ point particles \ay{with equal mass} and} initial positions $\{X_i\}_{i=1}^N$, \jl{which are assumed to} be independent, identically distributed random variables with density $f_x$. The initial velocities of the particles \jl{are denoted by} $\{V_i\}_{i=1}^N$ \jl{and are also assumed to} be independent, identically distributed random variables with continuous, bounded density $f_v$.  In particular, almost surely, no two particles have the same position or velocity.
%
Between collisions, particles undergo free flight.  

We start by proving that all collisions occur within a finite time $T^{(N)}$, as long as $N$ remains finite. 
\jld{Since an elastic collision corresponds to an exchange of velocities between the colliding particles, switching particle labels during collisions turns each labelled trajectory into a straight line.} \seth{ This point of view gives a \jl{simple} way of calculating all the collision times, past and present, in terms of the initial data only. \jl{Indeed, l}et \jld{$\ell_i(t) = X_i + tV_i$ denote the labelled trajectory of the $i$th particle, and} $\tau_{i,j}$ denote the \jl{time of} intercept of the paths $\ell_i(\cdot)$ and $\ell_j(\cdot)$ of particles $i$ and $j$.  Then, 
$$\tau_{i,j} = \frac{X_j-X_i}{V_i - V_j}$$
and the set of all the line intersection times, $\{ \tau_{i,j}: 1\le i \neq j \le N\}$ is in \jld{one-to-one} correspondence with the set of all collision times of the particles undergoing elastic collisions if we consider both positive and negative times.  We remark, since the distributions of the initial positions and velocity are continuous, these times are distinct almost surely.  Also, although these intersection times are not independent random variables, for example $\tau_{1,2} $ and $\tau_{1,3}$ both depend on $(X_1,V_1)$, the\jld{y} are `exchangeable' in \jld{the sense} that if $\pi$ is a permutation of $\{1,\ldots, N\}$, then $\{\tau_{i,j}: 1\leq i\neq j \leq N\} \mathop{=}\limits^d \{\tau_{\pi(i),\pi(j)}: 1\leq i \neq j\leq N\}$.}
\jld{The intersection times, $\tau_{i,j}$, are examples \jl{of} ratio distributions.
One can construct an integral representation of the distribution of such a ratio, although evaluating the integral in closed form is often difficult.  However, for example, if $\{X_i\}_{i=1}^N$ and $\{V_i\}_{i=1}^N$ are all iid $N(\mu,\sigma)$, then any $\tau_{i,j}$ follows a Cauchy distribution.  }

\seth{\jld{A}s the collection of collision times consists of $\binom{N}{2}$ terms, after the final (random) time $T\jld{^{(N)}}:=\max_{i\neq j}\tau_{i,j}$, there will be no more collisions between particles.  After this time, the labels of the lines have been sorted according to their velocities: \jld{i}n a vertical cross-section of the space time diagram, \jld{they are, from top to bottom,} in order from the largest to the smallest initial velocities.  \jld{Similarly,} after time $T^{(N)}$, all particles are arranged on the one-dimensional line in order of increasing speed, moving only by free flight since no further collisions can occur.
}

\seth{\jld{Let $t_i^{(N)}$ be the last time the $i$th particle interacts with any other particle.} Since the initial positions and velocities are independent and identically distributed, the random variables $\{t_i^{(N)}: 1\leq i\leq N\}$ are exchangeable, although not independent.  On the other hand, let $r_i^{(N)} := \max_{j: j\neq i} \tau_{i,j}$ denote the last time another line intersects $\ell_i(\cdot)$.  Then, as the last collision time of a particle must be the last intersection time of a line, we have $\{t_i^{(N)}: 1\leq i\leq N\} \mathop{=}\limits^d  \{r_i^{(N)}: 1\leq i\leq N\}$.  Moreover, as the initial positions and velocities are independent, $r^{(N)}_i$ equals any one of the $\{t_i^{(N)}: 1\leq i\leq N\}$ with equal probability.  In particular, 
$$P(r_i^{(N)} \leq x) = N^{-1}\sum_{j=1}^N P(t_j^{(N)}\leq x) = P(t_i^{(N)}\leq x),$$
 and so $t_i^{(N)}$ and $ r_i^{(N)}$ have the same distribution.  }
Also, 
$$T^{(N)} = \max_{1\leq i\leq N} r_i^{(N)}  = \max_{1\leq i\leq N}t_i^{(N)}.$$
\ay{Although not \jld{discussed} here, \jld{we remark that} it is also natural to consider, as an alternative, the minimum of the collision times. Given the symmetric distribution of $\tau_{i,j}$ it follows that $$\min_{1\le i < j \le N} \tau_{i,j} \mathop{=}\limits^d  -T^{(N)}.$$}

\subsection{\jld{Non-}elastic collisions}
\label{sec:inelastic_intro}
\ayd{Consider the more general, linear inelastic collision rule}
\begin{equation}
\begin{split}
v_i' &= (1-\epsilon)v_j \jld{+ \beta v_i,}\\
v_j' &= (1-\epsilon)v_i \jld{+ \beta v_j}.
\end{split}
\label{eq:coll_rule_general}
\end{equation}
Recall that $v_i,v_j$ ($v_i',v_j'$) are the pre (post) collision velocities of particles $i$ and $j$. \jld{The} parameters $\jld{\beta \ge\ } 0$ and $\epsilon < 1$ control conservation and dissipation of momentum and energy. \jld{In the case $\beta=\epsilon\ne0$, Eq. \eqref{eq:coll_rule_general} corresponds to inelastic collisions with coefficient of restitution $C_R = 1 - 2 \epsilon$}, which conserve momentum but not energy.
\jld{In what follows we choose $\beta = 0$ to reduce} the computational cost of \jld{the simulations, thereby making} the large numerical study from Sec. \ref{sec:sim} possible.  The simplicity of the collision rule \jld{obtained when $\beta = 0$ allows} for some analytic investigation \ayd{as well}. 
%
%
 \jld{Moreover, s}ince reversing the direction of time when a collision occurs amounts to changing $\epsilon$ into \ay{$1- \frac{1}{1-\epsilon} \simeq -\epsilon$} at leading order when $\epsilon$ is small, the distributions of collision times \jld{are, for small values of $\epsilon$, expected} to be nearly symmetric under the transformation $\cal T$ sending $t \to -t$ and $\epsilon \to -\epsilon$\jld{. This is confirmed by the simulations of Fig. \ref{fig:tau_diss}, which interestingly indicate that the symmetry under $\cal T$ is qualitatively observed, even for values of $\epsilon$ equal to 0.1.} \jld{We also note that, t}\ay{he transformation $\cal T$ can be used to understand properties of the minimum order statistics of the collision times through study of $t_i^{(N)}$ and $T^{(N)}$.} \jld{As already mentioned,} the pairwise collision times display a system size dependency which \jld{is} not present for elastic collisions. \jld{Together with the} broken symmetry, \jld{this effect} can be viewed as \jld{resulting from} the cooling (heating) of the system in the case of energy dissipating (generating) collisions\jld{: a} greater number of particles results in more collisions\jld{, which is compounded into increased} cooling (or heating).  
\begin{figure}[h!]
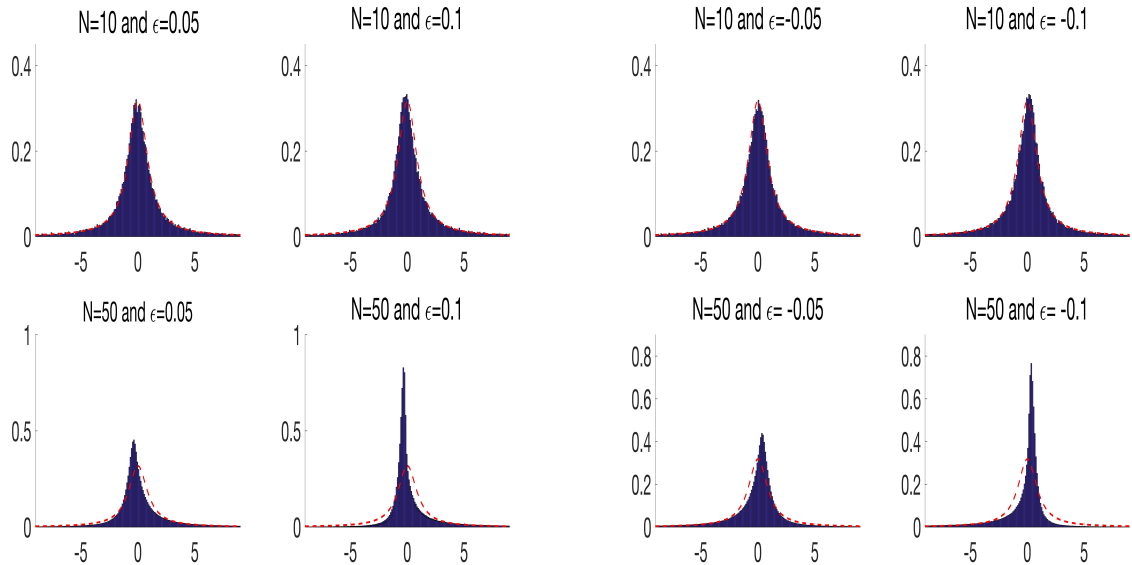

\begin{center}
\includegraphics[width = 3.2 in,height=3.2 in]{Pairwise_collisions_inelastic_diss.pdf}  \includegraphics[width = 3.2 in,height=3.2in]{Pairwise_collisions_inelastic_gen.pdf}
\caption{Histograms of \jl{collision} times \jl{$\tau_{i,j}$ for positive (left) and negative (right) values of $\epsilon$ in \jld{the non-}elastic collision rule (\ref{eq:coll_rule})} \jld{with $\beta = 0$}. The initial positions and velocities are sampled from a standard normal distribution.  \jl{For comparison, the densities} of collision times for elastic collisions \jl{are} shown in \ay{dashed} red. In this case, the elastic collision times follow a Cauchy distribution.}
\label{fig:tau_diss}
\end{center}
\end{figure}

\ay{For elastic collisions, the one-to-one correspondence of path intersection times and collision times made the argument for the existence of a final collision time straightforward.  For \jld{non-}elastic collisions, the existence of a final collision time is less \jld{obvious}.  Nonetheless, \jld{the proposition below indicates} there is almost surely a final collision time for a system of particles interacting with the collision rule from Eq. \eqref{eq:coll_rule} \jld{when} $\epsilon <1$.}

\begin{proposition}\label{prop:FTC_inelastic}
Suppose $N$ point particles with equal mass have initial positions $\{X_i\}_{i=1}^N$ and velocities $\{V_i\}_{i=1}^N$ \jld{that} are continuous random variables, so that $X_i \ne X_j$, $V_i\ne V_j$ for $i\ne j$ a.s.  If the particle\jld{s} interact \jld{through} the collision rule given in Eq. \eqref{eq:coll_rule} with $\epsilon \ne 0,\  \epsilon < 1 $, there is a final collision time almost surely.
\end{proposition}

\begin{proof}[Proof 
]
Suppose at the time of a collision, in addition to the change in velocity, the labels of particles are also switched.  Unlike \ayd{the case of} elastic collisions, the path of a particle \jld{in space-time} is not a straight line.  Instead, individual particles follow piecewise linear trajectories where the slope of \jld{each segment} changes by a factor of $1-\epsilon$ \jld{after each \ayd{intersection with another path}}.   Again, let $\ell_i(t)$ denote the position of the $i$th particle and $\tau_{i,\star_1}<\tau_{i,\star_2}<\dots$ denote the times when the path of particle $i$ intercepts the path of another particle. Then 
\begin{equation*}
\ell_i(t) = 
\begin{cases} 
X_i +t V_i & 0<t \le \tau_{i,\star_1} \\
(X_i+\tau_{i,\star_1} V_1) + (t-\tau_{i,\star_1}) (1-\epsilon) V_i & \tau_{i,\star_1} < t \le \tau_{i,\star_2} \\
\vdots 
\end{cases}
\end{equation*}

As was the case for elastic collisions, the collection of intersection times of the trajectories of particles (now piecewise linear) is in one-to-one correspondence with the collection of collision times.
\jld{Therefore,} for a system to have infinitely many collision times, there must be two paths which intersect infinitely many times.  \ay{Since the initial velocities are continuous random variables, it follows that for all $1\le i < j \le N$ and $k,m \in \{0,1,2,\dots\}$, $(1-\epsilon)^k V_i \neq (1-\epsilon)^m V_j$ almost surely  which ensures that any two paths which intersect do so by crossing one another a.s. }

We now assume the paths $\ell_1$ and $\ell_2$ intersect infinitely often and derive a contradiction. Since the paths cross, we may choose two successive intersection times, $s_1^{(2)}<s_2^{(2)}$, such that $\ell_1(t) > \ell_2(t)$ for $t\in (s_1^{(2)},s_2^{(2)})=\mathbb{S}^{(2)}.$  It may be the case that a third path also intersects $\ell_1$ and $\ell_2$ at time $s_1^{(2)}$ or $s_2^{(2)}$.  However, we may relabel the paths so that a ternary intersection does not alter the velocities of paths $\ell_1$ and $\ell_2$ on the open interval $\mathbb{S}^{(2)}.$

For example, suppose $\ell_3$ also intersects with $\ell_1$ and $\ell_2$ at $s_1^{(2)}$ or $s_2^{(2)}$.  There are two possibilities.  Firstly, $\ell_3$ does not intersect $\ell_1$ or $\ell_2$ on the interval $(s_1^{(2)},s_2^{(2)}).$  Thus, $\ell_3$ will not alter the velocities of $\ell_1,\ell_2$ on this interval and so will not influence the time, $s_2^{(2)}$, when $\ell_1$ and $\ell_2$ intersect again.  

Secondly, $\ell_3$ could intersect either $\ell_1$ or $\ell_2$ at a second time $s\in (s_1^{(2)},s_2^{(2)}).$  For instance, suppose $\ell_3$ intersects with $\ell_1$ at times $s$ and $s_2^{(2)}$, and that $\ell_3$ does not intersect $\ell_2$ on the interval $(s,s_2^{(2)}).$ We can then relabel the paths and times so that, $\ell_1$ and $\ell_2$ (formerly $\ell_1$ and $\ell_3$), intersect at successive times $s_1^{(2)}$ and $s_2^{(2)}$ (formerly $s$ and $s_2^{(2)}$).  Then on the interval $(s_1^{(2)},s_2^{(2)})$, the path $\ell_3$ (formerly $\ell_2$) intersects with neither $\ell_1$ nor $\ell_3$.  Other possibilities are handled similarly. As a result of this choice, on the closed interval $\overline{\mathbb{S}^{(2)}}$, any path which crosses both $\ell_1$ and $\ell_2$ does so an equal number of times.

Let us return to case where $\ell_1$ and $\ell_2$ intersect at successive times $s_2^{(1)}$ and $s_2^{(2)}$ and $\ell_1(t)>\ell_2(t)$ for $t\in\mathbb{S}^{(2)}$.  Furthermore, assume these paths are chosen so that there is no other path which intersects both $\ell_1$ or $\ell_2$ at $s_1^{(2)}$ or $s_2^{(2)}$ and either $\ell_1$ or $\ell_2$ at another time in the interval $\mathbb{S}^{(2)}.$  Let
$$
u_i^+ = \lim_{h\to 0^+}\frac{\ell_i(s_1^{(2)}+h)-\ell_i(s_1^{(2)})}{h}, \hspace{2mm} i=1,2\\
$$
be the velocities of particles $1$ and $2$ immediately after their collision at $s_1^{(2)}.$  Let 
$$
U_i^- = \lim_{h\to 0^-}\frac{\ell_i(s_2^{(2)}+h)-\ell_i(s_2^{(2)})}{h}, \hspace{2mm} i=1,2\\
$$
be the velocities of particles $1$ and $2$ immediately before their collision at $s_2^{(2)}.$  Since $\ell_1(t) > \ell_2(t)$ for $t\in \mathbb{S}^{(2)}$ and the paths cross a.s., it follows that 
\begin{equation}
u_1^+>u_2^+ \hspace{2 mm} \text{ and } \hspace{2 mm} U_1^- < U_2^- \hspace{2mm} \text{ a.s}.
\end{equation}  We note that the velocities of both paths $1$ and $2$ must be of the same sign since $U_i^-$ is related to $u_i^+$ by the equation $U_i^-=(1-\epsilon)^{k_i}u_i^+$ where $(1-\epsilon)^{k_1}>0$ and $k_i$ is the number of additional path intersections of $i$ during $\mathbb{S}^{(2)}$.  

Consider the path intersections involving $\ell_1$ on the interval $\mathbb{S}^{(2)}$ (Fig. \ref{FIG:Path_Intersections_Diagram}).  There are two possibilities.  (1) A path, (shown as a dotted, red line in Fig. \ref{FIG:Path_Intersections_Diagram}), crosses both $\ell_1$ and $\ell_2$ (the dashed, blue lines in Fig. \ref{FIG:Path_Intersections_Diagram})  $k$ times, and alters the velocities of both particles by a multiplicative factor of $(1-\epsilon)^k.$ 
Alternatively, (2) \jld{a path} may only intersect $\ell_1$ leaving the velocity of particle 2 unchanged (the solid black line in Fig. \ref{FIG:Path_Intersections_Diagram}).

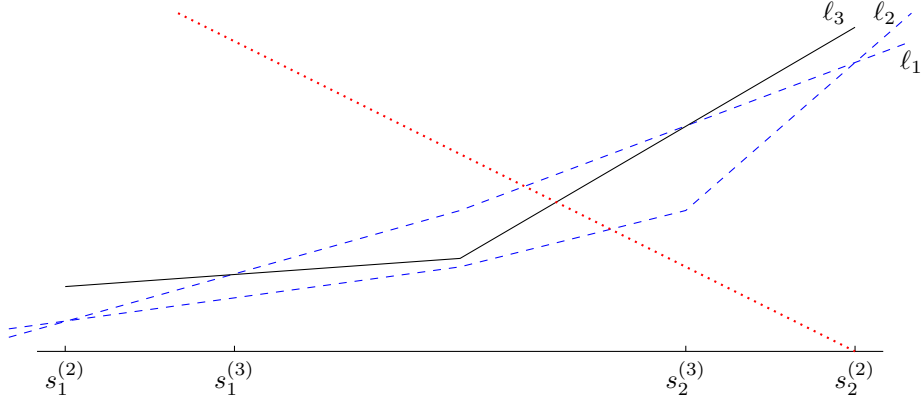
\begin{figure}[h!]
\begin{center}
\begin{tikzpicture}[scale=0.75]
\def\stwo{7};
\draw[dashed] (-\stwo,-3) -- (-\stwo,-2.9); 
\draw[dashed] (\stwo,-3) -- (\stwo,-2.9);
\draw (-\stwo,-3) node[anchor=north] {$s_1^{(2)}$};
\draw (\stwo,-3) node[anchor=north] {$s_2^{(2)}$};

\def\sthree{4};
\draw[dashed] (-\sthree,-3) -- (-\sthree,-2.9); 
\draw[dashed] (\sthree,-3) -- (\sthree,-2.9);
\draw (-\sthree,-3) node[anchor=north] {$s_1^{(3)}$};
\draw (\sthree,-3) node[anchor=north] {$s_2^{(3)}$};

\draw[solid,black] (-7.5,-3) --(7.5,-3);


\draw[solid, blue,dashed] (-\stwo-1,-2.75) -- (0,-0.5) -- (\stwo+1,2.5);
\draw (\stwo+1,2.5) node[anchor=north] {$\ell_1$};

\draw[solid, blue,dashed] (-\stwo-1,-2.6) -- (0,-1.5) --  (4,-0.5) -- (\stwo+1,3);
\draw (\stwo+0.9,3) node[anchor=east] {$\ell_2$};

\draw[solid, black] (-7,-1.85) -- (0,-1.35) -- (7,2.75);
\draw (7,3) node[anchor=east] {$\ell_3$};


\draw[red,dotted,thick] (-5,3) -- (7,-3);
\end{tikzpicture}
\caption{Paths $\ell_1$ and $\ell_2$ are shown as dashed, blue lines. The dotted red line serves as an example of a path which intersects both $\ell_1$ and $\ell_2$.  A path which intersects $\ell_1$ twice without hitting $\ell_k$ (shown as a solid black line) implies the existence of another `nested' path, $\ell_{k+1}$, also crossing $\ell_1$ twice but between the times, $s_1^{(k)}<s_2^{(k)}$, when $\ell_{k}$ hits $\ell_1$.  
}
\label{FIG:Path_Intersections_Diagram}
\end{center}
\end{figure}

If only $k_1$ intersections of the first type occur, then $$U_1^-=u_1^+(1-\epsilon )^{k_1} > u_2^+(1-\epsilon )^{k_1}=U_2^-$$ which contradicts the assumption that $U_1^- < U_2^-.$  \ayd{Indeed, $\ell_1$ and $\ell_2$ must experience a different number of intersections, $k_1$ and $k_2$ respectively, so that $$U_1^- = (1-\epsilon)^{k_1}u_1^+ < (1-\epsilon)^{k_2} u_2^+ = U_2^-.$$ The details of the ordering of $k_1$ and $k_2$ depend on both the sign of $u_1^+,u_2^+$ and on the sign of $\epsilon.$ However, we may assume without loss of generality that $k_1 > k_2 \ge 0.$}

Thus, there must be a path which intersects $\ell_1$ at least twice without intersecting $\ell_2$.  Call it $\ell_3$. Again, we may choose successive intersection times, $s_1^{(3)} < s_2^{(3)}$ in $\mathbb{S}^{(2)},$ such that $\ell_1(t)>\ell_3(t)$ and repeat the preceding argument (Fig. \ref{FIG:Path_Intersections_Diagram}).  

As such, the existence of a $k$th path with successive intersections of $\ell_1$ at times $s_1^{(k)}<s_2^{(k)}$ then implies the existence of a $(k+1)$th path which has successive intersections with $\ell_1$ on a subinterval $\mathbb{S}^{(k+1)}=(s_1^{(k+1)},s_2^{(k+1)})\subset (s_1^{(k)},s_2^{(k)})=\mathbb{S}^{(k)}$.  By construction, $\ell_2,\dots,\ell_k$ cannot collide with $\ell_1$ on $\mathbb{S}^{(k+1)}$.  Through induction we reach a contradiction as this requires an $(N+1)$th path in a system with only $N$ paths.

\end{proof}

 \jl{Like in the system with elastic collisions discussed above, collisions stop once particles become sorted in order of increasing velocity. However, the distribution of velocities of the particles \ay{\jld{is} no longer independent of collisions, but instead} depends on the number of collisions \jld{that} each particle experiences.  }

\section{Results for Elastic Collisions}\label{sec:results}
\subsection{Final Collision Time for a Single Particle}
The first theorem concerns the setting where $E|X_1|<\infty$.  In this case, the \ay{initial} \jl{position of the} center of mass of the system of particles converges to a constant value a.s. as $N$ tends towards infinity.
\begin{theorem}\label{FTC_particle1}
Let $t_i^{(N)}=\max_{i\ne j} \tau_{i,j}$ be the final path intersection time for a single particle in a system of $N$ particles with iid initial position $\{X_i\}_{i=1}^N$ \jl{with density $f_x$} and iid initial velocities $\{V_i\}\jld{_{i=1}^N}$ \jl{with density $f_v$}. Suppose $E|X_1| < \infty$ \jl{and that} $f_v$ is continuous and bounded. \jl{Then, f}or $\mu > 0$,
$$\lim_{N\to \infty} P\bigg(\frac{t_i^{(N)}}{N} < \mu\bigg) = \int_{\mathbb{R}^2}f_x(X)f_v(V)e^{-\frac{C(X,V)}{\mu}}dX dV  $$
where $$C(X,V) = f_v(V)\int_\mathbb{R} |X-y|f_x(y) dy.$$
\end{theorem}
\jl{In other words, the period of time during which each point particle undergoes elastic collisions scales linearly with the size $N$ of the system, as long as the position of the center of mass of the system remains finite. } \ay{For example, if the positions and velocities are iid $U[-1/2,1/2]$, then $$
\lim_{N\to \infty} P\bigg(\frac{t_i^{(N)}}{N} < \mu\bigg)  =\begin{cases} e^{-\frac{1}{4\mu}} \int_{-1/2}^{1/2} e^{-x^2/\mu} dx & \mu >0 \\
0 & \mu\le 0. \end{cases} $$
}

Theorems \ref{FTC_particle2} and \ref{FTC_particle3}
\seth{ consider when $X_1$ does not have mean, or}\ayd{,} when the \ay{initial }position of
the center of mass \jl{of the system of particles} does not converge \seth{a.s.} as $N$ tends towards infinity.  In this case, we require different rescalings of $t_i^{(N)}$ depending on the behavior of the tail of the distribution of $X_1.$ 

\begin{theorem}\label{FTC_particle2}
Let $t_i^{(N)}$ and $f_v$ be as in Thm. \ref{FTC_particle1}.  Suppose $\{X_i \}_{i=1}^N$ are iid with density $f_x(X) = \frac{C_\alpha}{1+|x|^{1+\alpha}}$ for some $\alpha \in (0,1)$. \jl{Then, f}or $\mu > 0,$
$$\lim_{N\to\infty} P\bigg( \frac{t_i^{(N)}}{N^{1/\alpha}} < \mu \bigg) = 
\int_{\mathbb{R}}  f_v(V) \exp\bigg( -\frac{C(V)}{\mu^\alpha}\bigg)  dV  $$
where 
$$C(V) = \frac{C_\alpha}{\alpha}\int_\mathbb{R} f_v(V+w)\frac{1}{|w|^\alpha} dw.$$
\end{theorem}

\begin{theorem}\label{FTC_particle3}
Let $t_i^{(N)}$ and $f_v$ be as in Thm. \ref{FTC_particle1}.  Suppose $\{X_i\}_{i=1}^N$ are iid Cauchy random variables.  \jl{Then, f}or $\mu > 0,$
$$\lim_{N\to \infty} P\bigg( \frac{t_i^{(N)}}{N \log N} \le \mu \bigg) = \int_\mathbb{R}f_v(V) \exp \bigg( - \frac{2f_v(V)}{\pi \mu} \bigg)  dV.$$
\end{theorem}
\jld{We remark, in passing, that} in the event $f_x(z)$ has different scalings as $z\to \pm\infty$, one can still construct asymptotic distributions of $t_i^{(N)}$.  The more slowly decaying tail will dominate the asymptotic behavior and the distribution will look similar to the result from Thm. \ref{FTC_particle2} or Thm. \ref{FTC_particle3} depending on the details of $f_x.$ 
\jl{The proofs of the above theorems are provided in Section \ref{sec:proof_tiN}.}

\subsection{Final Collision Time of the System}

We now turn to the scaling properties of $T^{(N)}$. Again, the first theorem presented concerns the situation when $X_1$ has a finite mean.  
\seth{We will in fact} require that the position distribution has at least a $3/2$-moment for technical reasons\jld{; details} of this proof are discussed in Sec. \ref{sec:proofs_TN}. 
\begin{theorem}\label{FTC_system}
Let $T^{(N)}=\max\limits_{1\le i < j \le N} \tau_{i,j}$ be the final  collision time of a collection of $N$ particles.  Suppose $E|X_1|^{3/2}<\infty$ and that $f_v$ is continuous and bounded. Then, for $t >0$,
$$\jld{\lim_{N\to \infty} } P\bigg(\frac{T^{(N)}}{\binom{N}{2}} \le t\bigg) \jld{=\ }  e^{-C/t },$$
where $$C = \int_{\mathbb{R}^2} |x-y|\, f_x(x)\, f_x(y)\, dx\, dy \cdot \int_\mathbb{R} f_v^2(v)\, dv.$$
\end{theorem}

This indicates that the final collision time of the system scales like the total number of particle pairs. 
The numerical simulation of Figure \ref{fig:TN_error_XN_VN} illustrates the convergence of the cumulative distribution of $T^{(N)}/\binom{N}{2}$ towards its asymptotic value on the interval 
$t \in (0,5)$ when both $\{X_i\}_{i=1}^N$ and $\{V_i\}_{i=1}^N$ are iid $N(0,1)$. 
The Silverman-Brown limit law used in the proof shown in Section \ref{sec:proofs_TN} indicates that the rate of convergence is $O(N^{-1})$, which is consistent with the numerical results. In order to numerically observe this rate of convergence, a total of $N^4$ trials had to be conducted to reconstruct the cumulative distribution of $T^{(N)}$.
\begin{figure}[h!]
\begin{center}
\includegraphics[width = 6 in]{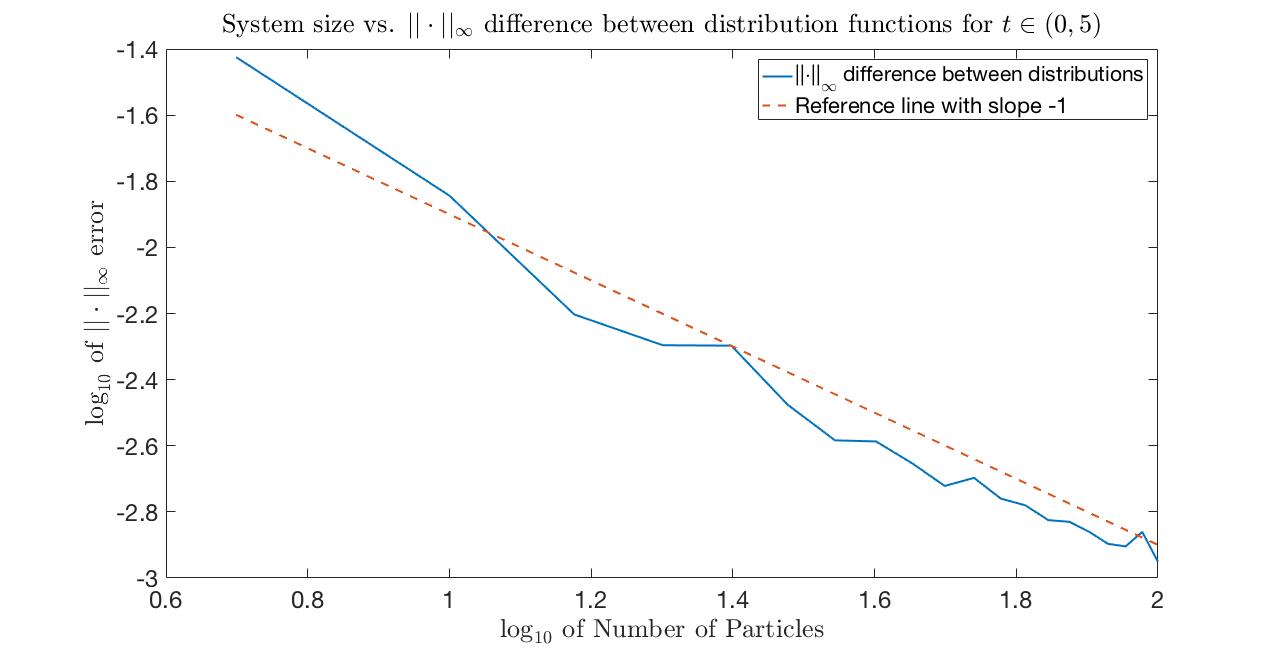}
\caption{The maximum difference between the empirical cumulative distribution function and theoretical asymptotic cumulative distribution function of $T^{(N)}/\binom{N}{2}$ over all values $t$ in the interval $(0,5)$ when both $\{X_i\}_{i=1}^N$ and $\{V_i\}_{i=1}^N$ are iid $N(0,1)$ with system sizes, $N$, between 5 and 100. For reference, the dashed line has slope negative one.}
\label{fig:TN_error_XN_VN}
\end{center}
\end{figure}

Two points are worth noting regarding the proof of Thm. \ref{FTC_system}. First, since
both $\{X_i\}_{i=1}^N$ and  $\{V_i\}_{i=1}^N$ are identically distributed,
the intersection times $\tau_{i,j}$ are also identically distributed. However, they are not independent. Nonetheless, if one assumes they are independent, then through the usual construction of the maximal order statistics of i.i.d. random variables, one recovers the same distribution as \jld{in} Thm. \ref{FTC_system} as $N\to\infty.$
Second, although the proof of the theorem requires that $E|X_1|^{3/2}<\infty$ for technical reasons, it is natural to wonder if $T^{(N)}/\binom{N}{2}$  is still converging weakly to some random variable if the requirements of the moment of $X_1$ are lessened.
 
\begin{proposition}\label{prop:tight}
Suppose $E|X_1| <\infty$ and $f_v$ is continuous and bounded, then the sequence of random variables $\{T^{(N)}/\binom{N}{2}\}_{N=2}^\infty$ is tight.
\end{proposition}

A proof of this proposition is given in Sec. \ref{sec:proofs_TN}.  Since the sequence, $\{T^{(N)}/\binom{N}{2}\}_{N=2}^\infty$ is tight, there is a random variable $T$ and subsequence $N_j$ such that $$\frac{T^{(N_j)}}{\binom{N_j}{2} }\Rightarrow T.$$  In the case of Thm. \ref{FTC_system}, the asymptotic distribution of $T^{(N)}/\binom{N}{2}$ is known. However, the tightness of the sequence $\{T^{(N)}/\binom{N}{2}\}$ requires only that $E|X_1|  <\infty$, a condition which also guarantees the constant $C$ of Thm. \ref{FTC_system} is well-defined.

\jld{It is natural to ask whether} the limiting distribution given in Thm. \ref{FTC_system} still hold\jld{s} if the requirement that $E|X_1|^{3/2}\jld{< \infty}$ is removed. \jld{Even though we do not provide a definite answer to this question, t}he following example \jld{gives} some insight. We simulated systems \jld{of increasing size $N$}, with random initial positions $\{X_i\}_{i=1}^N$ generated, via inverse sampling, from the density 
\begin{equation}
\label{eq:alpha_dens}
f_x(x) = \frac{9\cos(\pi/18)}{8\pi(1+|x|^{9/4})},
\end{equation}
for which $E|X_1|<\infty$ but $E|X_1|^\alpha$ is infinite for $\alpha \ge 5/4$. Random \jld{initial} velocities were sampled from a $N(0,1)$ distribution. Figure \ref{fig:TN_error_XAlpha_VN} shows a comparison between the empirical and theoretical cumulative distribution functions, similar to the example of Fig. \ref{fig:TN_error_XN_VN}, but for a larger range of values of $N$. The distribution of $T^{(N)}/\binom{N}{2}$ appears to be converging to the same theoretical limit as Thm. \ref{FTC_system}, but at a much slower rate, of order $O(N^{-0.35})$. This suggests that an argument for the convergence of $T^{(N)}$ to the same limiting distribution may exist \jld{for this example}. \seth{However, we show \jld{in the appendix} that the method \jld{of proof} we use cannot apply in this case.}
\begin{figure}[h!]
\begin{center}
\includegraphics[width = 6 in]{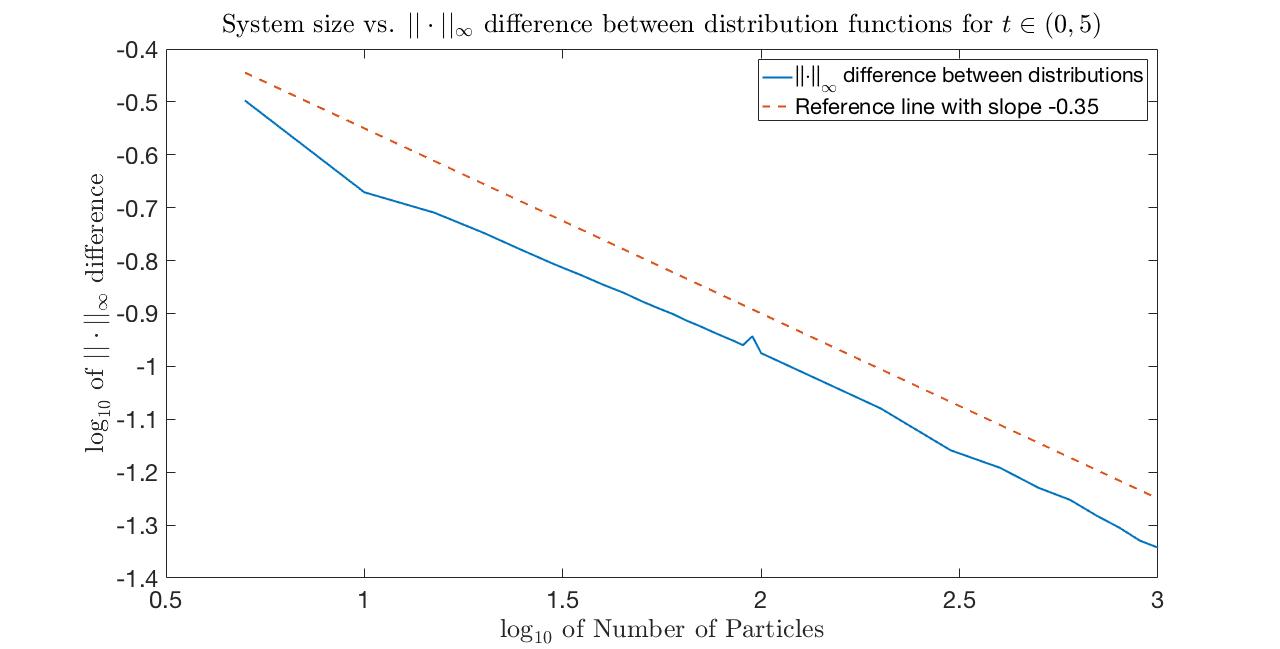}
\caption{
The maximum difference between the empirical cumulative distribution function and theoretical asymptotic cumulative distribution function of $T^{(N)}/\binom{N}{2}$ in the infinity norm over the interval $\mu\in(0,5)$\jld{. Here,} $\{X_i\}_{i=1}$ are \jld{iid} with density given in Eq. \eqref{eq:alpha_dens}, $\{V_i\}_{i=1}^N$ are iid $N(0,1)$\jld{, and the size of the} system, $N$, \jld{ranges} between 5 and 1000.  For reference, the dashed line has slope of -0.35.  For $N$ between 5 and 100, $N^4$ trials were conducted.  Over that interval the observed convergence rate was $N^{-0.35}$. As such, only $N^2$ trials were conducted \jld{for} systems sizes between 100 and 1000.}
\label{fig:TN_error_XAlpha_VN}
\end{center}
\end{figure}

\seth{We now discuss} asymptotic distributions for $T^{(N)}$ when the \jld{initial} position r.v. does not have a first moment.  Similar to Theorems \ref{FTC_particle2} and \ref{FTC_particle3} an alternate scaling is required.
\begin{theorem}
Let $T^{(N)}$ be the final collision time of a collection of $N$ particles.  Suppose the initial positions are iid with density $\displaystyle f_x(x) = \frac{C_\alpha}{1+\jld{|x|}^{1+\alpha}}$ for $\alpha \in (0,1)$ and that $f_v$ is continuous \jld{and bounded}.  Then for $t > 0$,
$$\jld{\lim_{N\to \infty} } P\bigg(\frac{T^{(N)}}{N^{2/\alpha}} \le t \bigg) \jld{=\ } e^{-C/t^\alpha}$$
where $$C = \frac{C_\alpha}{2\alpha} \int_{\mathbb{R}^2}f_v(V)\frac{f_v(V+w)}{|w|^\alpha} dw \, dV. $$
\label{FTC_system2}
\end{theorem}
\begin{theorem}
Let $T^{(N)}$ be the final collision time of a collection of $N$ particles.  Suppose the initial positions are iid Cauchy and that $f_v$ is continuous \jld{and bounded}.  Then for $t > 0$,
$$\jld{\lim_{N\to \infty} } P\bigg(\frac{T^{(N)}}{N^2\log N} \le t \bigg) \jld{=\ } e^{-C/t}$$
where $$C = \frac{\jld{2}}{\pi } \int_\mathbb{R}f_v^2(V) dV.$$
\label{FTC_system3}
\end{theorem}
Like Theorem \ref{FTC_system}, the limiting distributions in Theorems \ref{FTC_system2} and \ref{FTC_system3} are of Fr\'echet type.  Additionally, as was the case for Theorem \ref{FTC_system}, one recovers the same asymptotic distribution 
 if the pairwise collision times are assumed to be independent.  
The proofs of these theorems follow from the same asymptotic arguments employed in the proofs of Theorems \ref{FTC_particle2} and \ref{FTC_particle3} which are discussed in detail in Section \ref{sec:proof_tiN} \jld{and in the Appendix}.  

In summary,  the asymptotic properties of $T^{(N)}$ depend on the moment properties of the initial position r.v. In the simple case where the initial velocity r.v. has a continuous density and the initial position r.v. has a density of the form $f_x(x) = \frac{C_\alpha}{1+|x|^{1+\alpha}}$, \jld{we have the following results:} for $\alpha \in (0,1)$, Theorem \ref{FTC_system2} holds and $T^{(N)}$ scales like $N^{2/\alpha}$\jld{;} for $\alpha = 1$, Theorem \ref{FTC_system3} holds and $T^{(N)}$ scales like $N^2\log N$\jld{;} and for $\alpha >3/2$ Theorem \ref{FTC_system} holds  and $T^{(N)}$ scales like $N^2$.  \seth{As mentioned earlier,} the case for $\alpha \in (1,3/2]$ is open, but \jld{the} numerical stud\jld{y of Fig. \ref{fig:TN_error_XAlpha_VN} suggests} the result of Theorem \ref{FTC_system} still holds and $T^{(N)}$ scales like $N^2$. 

\section{Results for \jld{Non-}elastic Collisions}\label{sec:sim}


\jld{The presence of non-elastic collisions ($\epsilon \ne 0$) significantly \jld{affects} the asymptotic behavior of the system.} First, due in part to the heating/cooling effects observed in Sec. \ref{sec:inelastic_intro}, the asymptotic properties of $t_i^{(N)}$ and $T^{(N)}$ are \jld{modified.}
\jld{Importantly,} the \jld{initial} velocity distribution play\jld{s} a role, \jld{which was not the case for} elastic collision\jld{s, whose asymptotic behavior was solely determined by} the distribution of initial positions. 
\jld{Second}, the analytic framework \jld{we} used to determine the correct asymptotic behavior of $t_i^{(N)}$ and $T^{(N)}$ is \jld{no longer} viable. \jld{As a consequence, we do not have the means} to rigorously derive the asymptotic distributions \jld{of these quantities}
at this time.  Instead, we \jld{describe how their medians scale}
through a numerical study using molecular dynamics (MD) simulations. 

\jld{Since} under elastic collision\ayd{s}, neither the limiting distribution of $t_i^{(N)}/N$ nor that of $T^{(N)}/N^2$ \jld{has a} mean\jld{, we} assume \jld{this is also the case under non-elastic collisions. We thus investigate how the medians of $t_i^{(N)}$ and $T^{(N)}$ scale for large $N$.} 
We restrict our focus to $N(0,1)$ \jld{initial} position \jld{distributions}, and study separate cases when the initial velocities are iid $N(0,1)$, $N(1,1)$, $U(-1,1)$, and $U(0,2).$ Due to numerical constraints, we limit our \jld{investigation} to systems no larger than $N=500$ and \jld{to} values of $\epsilon \in (-5\times 10^{-3},5\times 10^{-3}).$ 
\jld{We generate $10^4$ simulations for each value of $N$ and $\epsilon$ studied, and use bootstrapping to estimate confidence intervals for $M_t$, the median of the final collision time of a particle, and for $M_T$, the median of the final collision time of the system.}

We begin with a review of the properties of the sampling distribution of the sample median, highlighting consistency with known properties \jld{of} the sample median under elastic collisions. We then discuss the effects of \jld{non-}elastic collisions on $t_i^{(N)}$ and $T^{(N)}$. 
Consider initial velocities and positions which are standard normal random variables. Under elastic collisions, \jld{we know that} the final collision time for a single particle scales linearly with the number of particles in the system (Thm. \ref{FTC_particle1}) while the final collision time for the system scales linearly with $N^2$ (Thm. \ref{FTC_system}).  Thus, the median of the final collision time of a particle is of the form $M_t = C_t N$ at leading order as $N\to \infty$ where $C_t$ is the value of $\mu$ satisfying the equation $$\frac{1}{2} = \frac{1}{2\pi} \int_{\mathbb{R}^2} \exp\bigg( -\frac{X^2+V^2}{2} - \frac{1}{2\pi \mu}e^{-v^2/2} \int_\mathbb{R} |X-y| e^{-y^2/2}dy \bigg)dX dV.$$
A numerical calculation indicates $C_t\approx 0.412.$  \jld{Similarly, t}he median of the final collision time of the system behaves like $M_T = C_T N^2$ at leading order as $N\to \infty$, where $C_T=\jld{1 / \big(}{\pi\log 4}\jld{\big)}$, which one can calculate directly from the asymptotic distribution \jld{provided} in Thm. \ref{FTC_system}.  

\begin{figure}[h!]
\begin{center}
\includegraphics[width = 5 in]{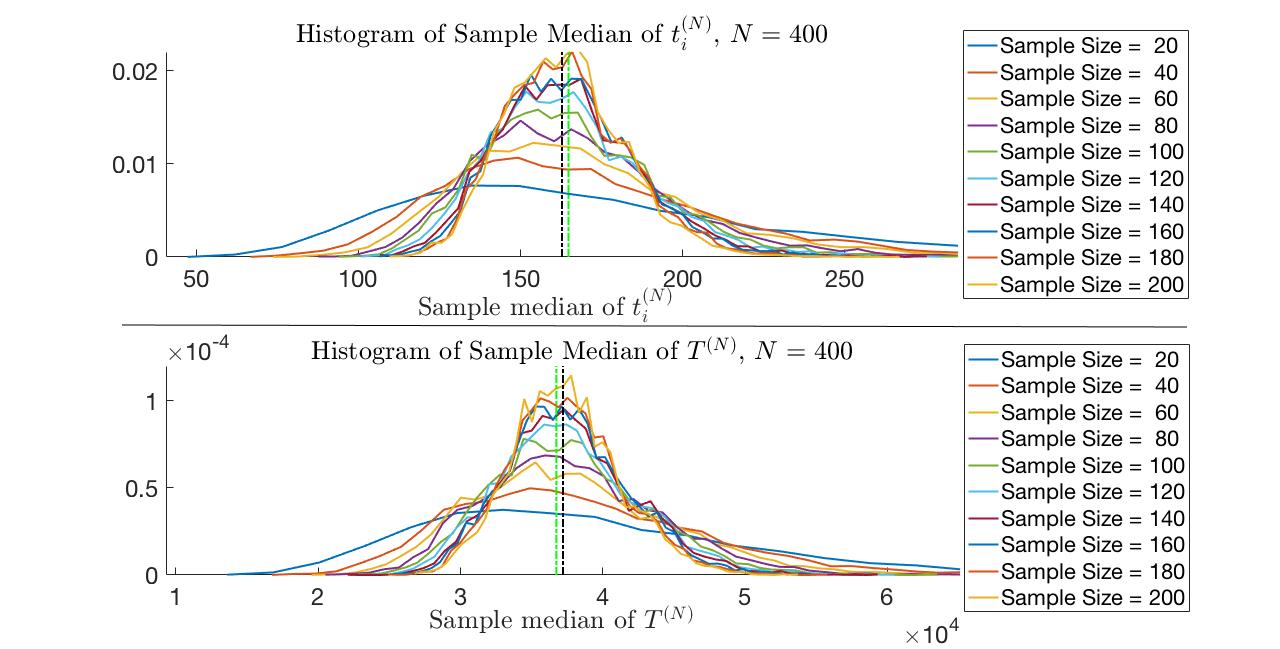}
\caption{Statistical properies of sample medians when initial position\jld{s} and velocit\jld{ies} are iid $N(0,1).$ (a) Numerical reconstruction of the sampling distribution of the median final collision time of a particle for increasing sample size in a system of 400 particle under elastic collisions.  The vertical dashed line \jld{in} black corresponds to a sample median generated from a sample of size $10^4.$ The line in green corresponds to the leading order behavior $M_t\approx C_tN.$
(b) Numerical reconstruction of the sampling distribution of the median final collision time of the system for increasing sample size in a system of 400 particle under elastic collisions.  The vertical dashed line \jld{in} black corresponds to a sample median generated from a sample of size $10^4.$ The line in green corresponds to the leading order behavior $M_T\approx C_TN^2.$}
\label{FIG:Density_sample_median}
\end{center}
\end{figure}

\begin{figure}[h!]
\begin{center}
\includegraphics[width = 5 in]{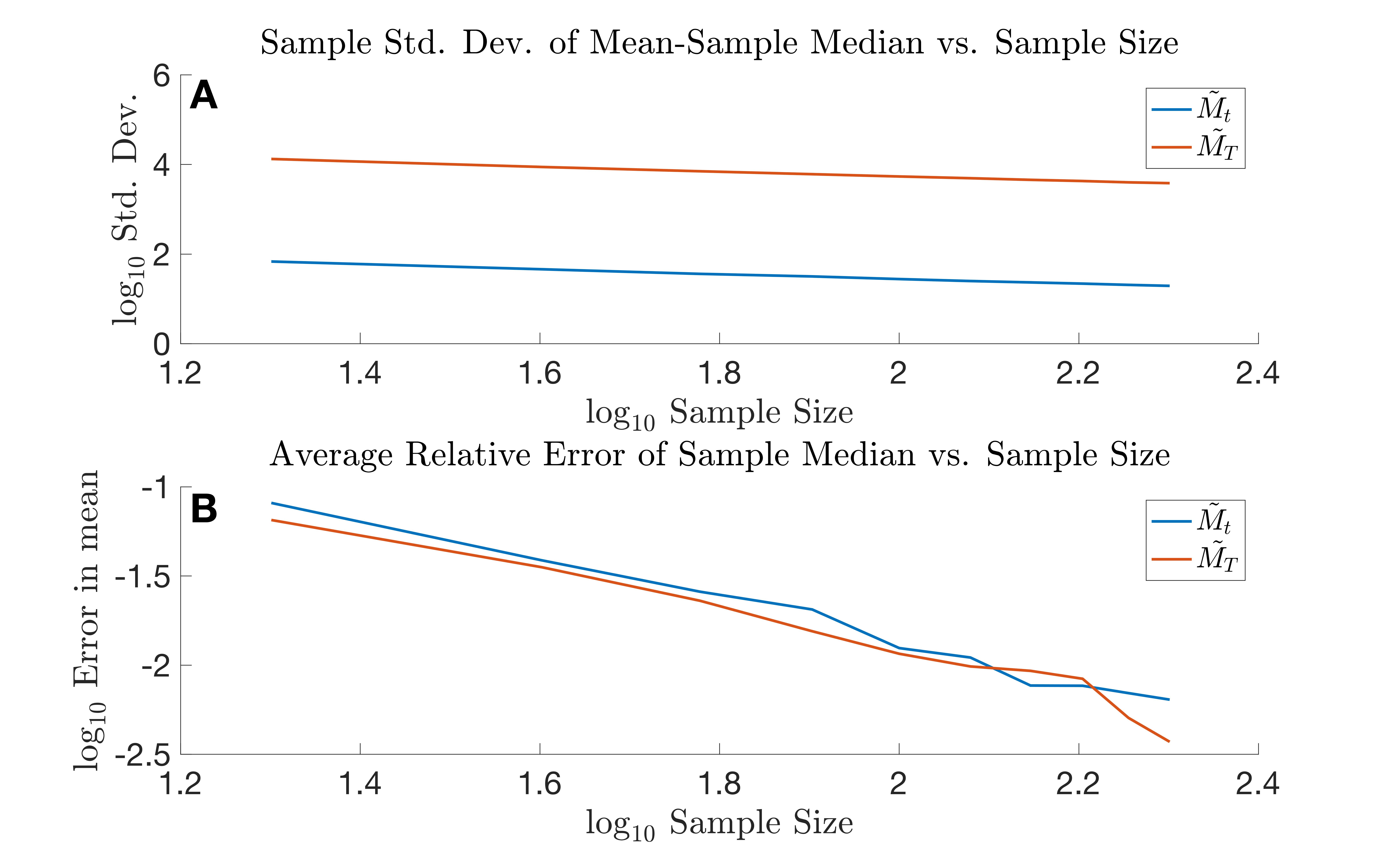}
\caption{(a) The sample standard deviation \jld{of} final collision time medians is $O(1/\sqrt{k})$ consistent with the results of \cite{Chu}. (b) The average relative error of the sample medians when compared to the true median is $O(1/k)$ and less than 1\% for both $\tilde{M_t}$ and $\tilde{M_T}$ when the sample size $k\ge 200.$  }
\label{FIG:sample_median_error}
\end{center}
\end{figure}

For a general random variable, $Y$, with density $f_Y$ and median $M_Y$, the sampling distribution of the sample median of $k$ iid samples from $Y$ converges to a normal distribution as the sample size approaches infinity.  The mean of the normal distribution is $M_Y$ and \jld{its} variance $\jld{1 / \big(}4k\, [f_Y(M_Y)]^2\jld{\big )}$ \cite{Chu}.  Histograms generated from sample medians under increasing sample sizes suggest this result extends to the sample median of both $t_i^{(N)}$ and $T^{(N)}$ (Fig. \ref{FIG:Density_sample_median}).
\jld{Indeed, the} $O(1/k)$ decay in the variance \jld{is} observed across all particle numbers simulated. Additionally, if one uses the sample median of $10^4$ simulations as an approximation for the true medians, $M_t$ and $M_T$, then the mean of the sample medians converges to the true median
at a rate which is $O(1/k)$ (Fig. \ref{FIG:sample_median_error}).  \jld{Specifically, f}or $k\ge 200$, the relative error between the mean sample median and \jld{the} true median approximated from $10^4$ simulations was less than 1\% for both $M_t$ and $M_T$ across all values of particle number simulated. 


\subsection{Median of $t_i^{(N)}$ under \jld{non-}elastic collisions}
	We \jld{first present results on} the behavior of the final collision time of a particle, $t_i^{(N)}.$ The median $M_t$ was approximated by selecting at random the final collision \jld{time} of one particle from each of $10^4$ simulations and taking the median of these values.   This process was repeated 100 times and the results were used to construct the 99\% confidence intervals for the true median. 
This construction of the confidence intervals assumes the approximately normal distribution of the sample medians observed for elastic collisions persists under \jld{non-}elastic collisions.  We use $M_t^{elas}$ to denote the median final collision time under elastic collisions and $M_t^{inel}$ to denote the median final collision time under \jld{non-}elastic collisions.
	
Before discussing the results, we propose an ansatz which relates path intersection times under elastic collisions to those under \jld{non-}elastic collisions.  Let $\epsilon =0$ and suppose that the path of particle 1 first intersects the path of particle 2\jld{, and} then proceeds to intersect the path of particle 3. Let $\tau_{1,j} = \frac{X_1-X_j}{V_j-V_1}$, $j=2,3$ be the path intersection times between particle 1 and particle\jld{s} $j=2,3$\jld{, and l}et $\tau_{1,3}'$ be \jld{the} time at which the paths of particles 1 and 3 intersect under \jld{non-}elastic collisions. If we assume particle 3 experiences a path intersection at approximately the same time that the paths \jld{of} particles 1 and 2 coincide, then
$$\tau_{1,3}' \approx \tau_{1,2} +\frac{(X_1+\tau_{1,2} V_1) - (X_3 + \tau_{1,2}V_3)}{(1-\epsilon)(V_3-V_1)} = \frac{1}{1-\epsilon}\tau_{1,3}.$$
\begin{figure}[h!]
\begin{center}
\includegraphics[width = 6 in]{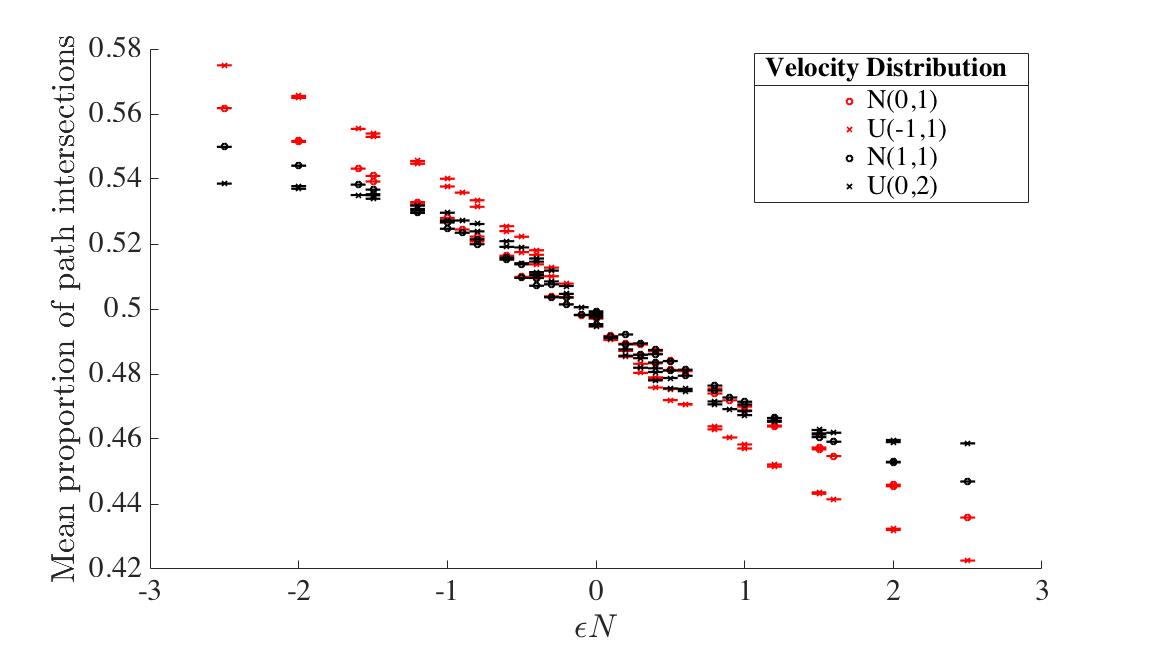}
\caption{The sample mean of the proportion of path intersections, $\overline{\alpha(\epsilon,N)},$ exhibits a clear dependence \jld{in} $\epsilon N$, which is influenced by the initial velocity distribution.}
\label{FIG:Path_Intersection_Proportions}
\end{center}
\end{figure}
\begin{figure}[h!]
\begin{center}
\includegraphics[width= 6 in]{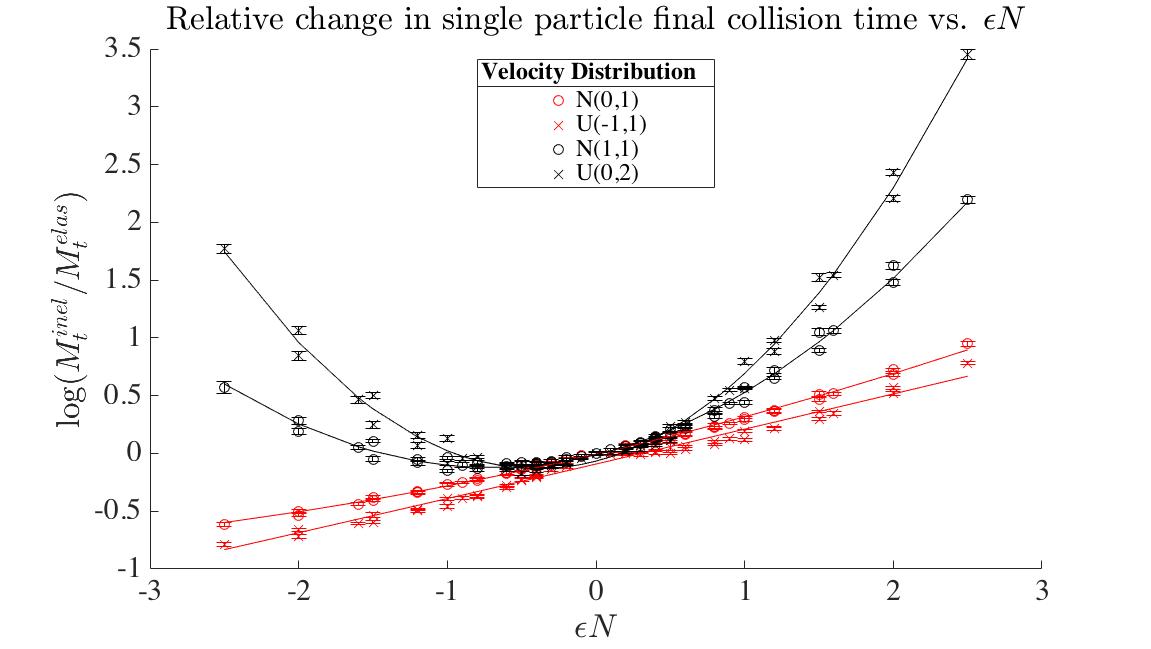}
\caption{Confidence intervals (95\%) of the relative change in the median single particle final collision time\jld{. This quantity} exhibit\jld{s} a \jld{clear} dependence on $\epsilon N$\jld{, which} is influenced by the velocity distribution.  The coefficients of the regression \ayd{$D_1(\epsilon N) +D_2(\epsilon N)^2$} applied to each of the four velocity distributions are \jld{displayed} in Table \ref{TABLE:Regression_Coefficient_t}.}
\label{FIG:Medians_inelastic_t}
\end{center}
\end{figure}
Thus, under \jld{non-}elastic collisions, the second path intersection time differs from the elastic path intersection time by a factor of $\frac{1}{1-\epsilon}.$ Generalizing this argument to the the $j^{th}$ interaction of particle 1, $$\tau_{1,j}' \approx \frac{1}{(1-\epsilon)^{j-1}} \tau_{1,j}.$$ This extension implicitly assumes that particle 1 has experienced an (approximately) equal number of path intersection as \jld{every} other particle with which it interacts.  Furthermore, the interactions should have occurred at approximately the same time.  As such, it is not reasonable to expect the ansatz to apply to the longest intersection times, i.e. those which correspond to the final collision time of the system.  
However, \jld{assuming that} one can extend this ansatz to the final collision \jld{time} of a given particle, then \jld{non}-elasticity alters $t_i^{(N)}$ by \jld{a} multiplicative factor of the form $$\frac{1}{(1-\epsilon)^{N\alpha(\epsilon,N)}} = e^{-\log(1-\epsilon) \jld{N} \alpha(\epsilon,N)}$$
where $\alpha(\epsilon,N)$ represents the proportion of path intersections \jld{experienced by particle} $i$. 
\jld{Since f}or $|\epsilon | \ll 1$, the multiplicative correction is approximately $\exp(\epsilon N \alpha(\epsilon,N))$\jld{, we expect} $\log(M_t^{inel} \jld{/} M_t^{elas})$ \jld{to} depend on $\epsilon N$ \jld{at} leading order.  
Furthermore, \jld{numerical simulations indicate that} the mean proportion of path intersections, $\overline{\alpha(\epsilon,N)}$, is a function of $\epsilon N$ (Fig. \ref{FIG:Path_Intersection_Proportions})\jld{, with leading order contribution equal to $1/2$. This}
is consistent with the sorting process interpretation of elastic collisions\jld{: e}ach particle experience\jld{s} $N-1$ path intersections\jld{, including both} positive and negative \jld{times; g}iven the symmetric distribution of $\tau_{i,j}$ about zero, half of these path intersections \jld{are expected to} occur for $t>0$ on average. \jld{Therefore, one would expect} $\log(M_t^{inel}/M_t^{elas})\sim \frac{\epsilon N}{2}$ \jld{at} leading order.
\jld{W}hile the general structure of this relationship is correct, including the dependence on $\epsilon N$, \jld{a} numerical reconstruction of $\log(M_t^{inel}/M_t^{elas})$ indicates this assumption is \jld{quantitatively inaccurate} (Fig. \ref{FIG:Medians_inelastic_t}). \jld{Indeed, a}pplying a regression of the form \ayd{$D_1 (\epsilon N) + D_2 (\epsilon N)^2$} to  the median final collision results indicates that $\ayd{D_1}\approx 1/3$ rather than the value of $1/2$ predicted by the ansatz (Table \ref{TABLE:Regression_Coefficient_t}).
The difference is likely due to the assumption that particle\jld{s} experience equal path intersections\jld{, which is only reasonable for small times.}
A \ayd{more rigorous} study of the \jld{number of path intersections experienced by colliding particles}
could shed some light on this issue.

\begin{table}[ht]
\centering
\begin{tabular}{|c|| c | c |}
\hline
Distribution &  $D_1$ & $D_2$ \\
\hline
$N(0,1)$   & 0.299364165261241  & 0.023741992291113 \\
\hline
$U(-1,1)$& 0.299564298338677  & -0.0284407509187195 \\
\hline
$N(1,1)$   & 0.315410496892726 & 0.221319632974651 \\
\hline
$U(0,2)$  & 0.333999962446071  & 0.403837320001947 \\
\hline
\end{tabular}
\caption{
\jld{Coefficients of} the regression \ayd{$D_1 \epsilon N + D_2(\epsilon N)^2$} applied to $\log(M_t^{inel}/M_t^{elas})$.  The linear \jld{term has a slope \ayd{$D_1$}} approximately equal \jld{to $1/3$} across all velocity distributions studied.  \jld{T}he quadratic coefficient \ayd{\jld{$D_2$}} is much larger for the distributions centered about \jld{$1$}.}
\label{TABLE:Regression_Coefficient_t}
\end{table}
\jld{T}his study \jld{highlights an important} difference \jld{between} elastic \jld{and non-elastic} collisions. \jld{While} the velocity distribution play\jld{s} no role in determining the scaling of $t_i^{(N)}$ \jld{when $\epsilon = 0$, it clearly affects the value of $C_2$ when $\epsilon \ne 0$, although the effect is}
weaker for those velocity distributions \jld{that} are symmetric about zero.

\subsection{Median of $T^{(N)}$ under \jld{non-}elastic collisions}

\jld{The median} $M_t$ \jld{considered in the previous section} captures the final collisions time of a randomly selected particle within the system, \jld{and therefore} reflects the \jld{behavior} of `typical' particle\jld{s}.  \jld{On the other hand, t}he final collision time of the system,
$T^{(N)}$, arises from \jld{those particles that} exhibit the largest final collision times,
whose behavior is \jld{thus} atypical.  As such, the ansatz proposed in the previous section cannot be expected to adequately describe the behavior of $T^{(N)}$.  Nonetheless, \jld{the numerical results of Fig. \ref{FIG:Medians_T_inelastic} confirm that} $\log(M_T^{inel}/M_T^{elas})$ still depends on $\epsilon N$ and \jld{on} the velocity distribution. \jld{Table \ref{TABLE:Regression_Coefficient_T} shows the coefficients obtained by applying  a regression of the form \ayd{$D_1(\epsilon N) + D_2(\epsilon N)^2$} to this quantity. As before, the effect on the quadratic term is stronger for distributions centered about $1$.}
\begin{table}[ht]
\centering
\begin{tabular}{| c ||  c | c|}
\hline
Distribution & $D_1$ & $D_2$ \\
\hline
$N(0,1)$ &  0.316098755619495 &	 0.353122072713743 \\
					\hline
$U(-1,1)$ & 0.294877488559631&  0.440060818253797		 \\
\hline
$N(1,1)$ & 0.270061962096432 & 1.65459638500048		 \\ 
\hline
$U(0,2)$ & 0.184011367253361 &  2.53476516441542 \\
\hline
\end{tabular}
\caption{\jld{Coefficients of} the regression \ayd{$D_1 \epsilon N + D_2(\epsilon N)^2$} applied to $\log(M_T^{inel}/M_T^{elas})$.  As in the study of $t_i^{(N)}$, the quadratic \jld{term} is much larger for the distributions centered about \jld{$1$}. 
}
\label{TABLE:Regression_Coefficient_T}
\end{table}

\begin{figure}[h!]
\begin{center}
\includegraphics[width=6 in]{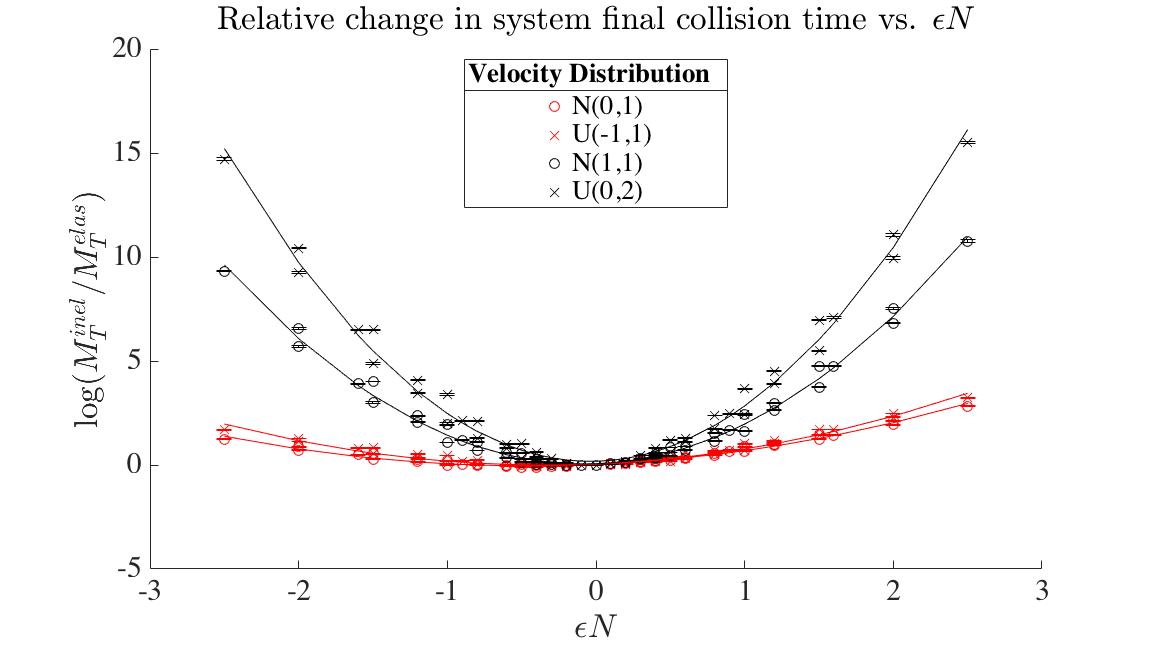}
\caption{Confidence intervals for $\log(M_T^{inel}/M_T^{elas})$.  \jld{The coefficients of the regression \ayd{$D_1(\epsilon N) +D_2(\epsilon N)^2$} applied to each of the four velocity distributions are \jld{displayed} in Table \ref{TABLE:Regression_Coefficient_T}.}}
\label{FIG:Medians_T_inelastic}
\end{center}
\end{figure}


\begin{figure}[h!]
\begin{center}
\includegraphics[width=6 in]{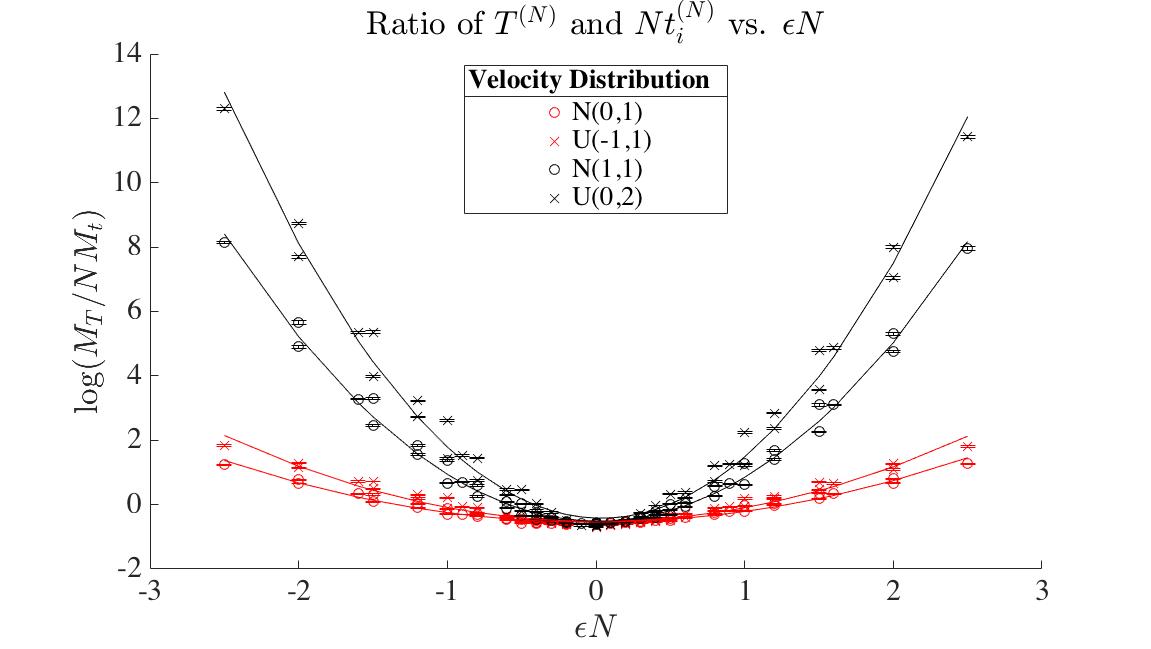}
\caption{Confidence intervals for $\log(M_T^{inel}/NM_t^{elas})$ are symmetric about $\epsilon N = 0.$  \jld{The coefficients of the regression} \ayd{$D_0+D_1(\epsilon N) +D_2(\epsilon N)^2$} \jld{applied to each of the four velocity distributions are displayed} in Table  \ref{TABLE:Regression_Coefficient_TNt}. }
\label{FIG:Medians_TN_inelastic}
\end{center}
\end{figure}

\jld{T}he asymmetry of $\log(M_T^{inel}/M_T^{elas})$ about $\epsilon N = 0$ \jld{observed in Fig. \ref{FIG:Medians_T_inelastic}} is due to the \jld{linear dependence of $t^{(N)}$ on $\epsilon N$. We checked that} the ratio $M_T/N M_t$ behaves like $e^{O(1)+D_2(\epsilon N)^2}$(\jld{Figure \ref{FIG:Medians_TN_inelastic} and} Table \ref{TABLE:Regression_Coefficient_TNt}).  The $O(1)$ term arises from differences in the coefficients \jld{$C_t$ and} $C_T$ \jld{that appear} in the scaling of the medians, $M_t \sim \jld{C}_t N$ and $M_T \sim C_T N^2$, under elastic collisions.

\begin{table}[h!]
\centering
\begin{tabular}{| c || c | c | c|}
\hline
Distribution & $D_0$ & $D_1$ & $D_2$ \\
\hline
$N(0,1)$ &-0.542655411094044 &
0.0171180888904638 &
0.312943997391057 \\
\hline
$U(-1,1)$ & -0.517690994670863 &
-0.00466592391659517& 
0.424631130596750\\
\hline
$N(1,1)$ & -0.512993538107585 &
-0.0466014260783037 &
1.40886890163564	\\
\hline
$U(0,2)$ & -0.419236799598809 & 
-0.153106277504774 & 
2.05627305470200 \\
\hline
\end{tabular}
\caption{\jld{Coefficients of} the regression $D_0 + D_1 \epsilon N + D_2(\epsilon N)^2$ applied to $\log(M_T^{inel}/NM_t^{elas})$.  The values of $C_1$ are much smaller than in the previous tables. 
}
\label{TABLE:Regression_Coefficient_TNt}
\end{table}

\subsection{Discussion of \jld{Non-}elastic Collisions}

The rescaling of velocities by \jld{the} factor $(1-\epsilon)>0$ \jld{each time} a collision \jld{occurs} results in an increase (decrease) in kinetic energy when $\epsilon <0$ ($\epsilon >0)$.  While the overall increase (decrease) in the speed of particles is likely to increase (decrease) the rate at which collisions occur, surprisingly, the final collision times do not respond accordingly.  For $|\epsilon N| \ll 1$, both $t_i^{(N)}$ and $T^{(N)}$ are decreasing functions of $\epsilon N$, consistent with the heating/cooling effects of the \jld{non-}elastic collisions.  However, as $|\epsilon N|$ moves away from zero, $T^{(N)}$ increases until it reaches values larger than that of elastic collisions. 
 \ayd{In the present simulations, this behavior is independent from the value, $0$ or $1$, of the average initial particle velocity. However, under non-elastic collisions, $t_i^{(N)}$ increases to values larger than that of elastic collisions only when the average initial velocity is $1$.}

\jld{The numerical investigations discussed in this section lay out a possible path towards a} rigorous construction of the asymptotic \jld{behavior} of the final collision times \jld{by identifying three central questions: (i) d}oes the proposed ansatz fail to accurately capture the effects of \jld{non-}elastic collisions due to differences in path intersection \jld{counts} and if so, to what extent? \jld{(ii)} How do properties of the \jld{initial} velocity distribution affect \jld{the} scaling of the (median) collision times? And \jld{(iii)}, what aspects of the collision dynamics give rise to the observed increase in the median collision time? This \jld{last question} is perhaps most difficult to answer \jld{for} $T^{(N)}$, which is influenced largely by rare events (large collision times). We believe connections with the sorting processes outlined in Sections \ref{sec:sorting} \jld{and} \ref{sec:inelastic_intro} may be able to offer some insight, \jld{although a preliminary} numerical investigation \jld{(which is beyond the scope of this article)} has proven inconclusive to this point.  


\section{Proofs of theorems}\label{sec:proof}
\subsection{Proofs of Theorems Regarding $t_i^{(N)}$}\label{sec:proof_tiN}
We begin the proofs concerning Thms. \ref{FTC_particle1} - \ref{FTC_particle3} by \seth{deriving a formula for} the exact distribution of $t_i^{(N)}=\max_{j\ne i} \tau_{i,j}$. Using independence of $\{(X_i,V_i)\}_{i=1}^N$ and by conditioning on the initial velocity and position of particle $i$, we can express the conditional probability of the event $\{t_i^{(N)}< \mu_N | (x_i,v_i)=(X,V)\}$ as follows.
\begin{equation}
\begin{split}
P(t_i^{(N)} < \mu_N | (x_i,v_i) = (X,V)) &= P(\tau_{i,1} < \mu_N, \dots \jld{,} \tau_{i,i-1} < \mu_N, \tau_{i,i+1}<\mu_N  \jld{,} \dots \jld{,} \tau_{i,N}< \mu_N |(x_i,v_i)=(X,V)) \\
&=[P(\tau_{i,1} < \mu_N|(x_i,v_i)=(X,V))]^{N-1} \\
&=\bigg[P\bigg(\frac{x_1-X}{V-v_1} < \mu_N \bigg)\bigg]^{N-1} .
\end{split}
\label{eq:cond_prob}
\end{equation}
Integrating Eq. \eqref{eq:cond_prob} \jld{with respect to} the distribution of $(X,V)$ yields the unconditional probability of $\{t_i^{(N)} < \mu_N\}$,
\begin{equation}
\begin{split}
P(t_i^{(N)} < \mu_N) &= \int\int f_x(X)f_v(V)P(\tau_{i,j} < \mu_N | (x_i,v_i) = (X,V)) dX dV \\
&=\int\int f_x(X)f_v(V)\bigg[P\bigg(\frac{x_1-X}{V-v_1} < \mu_N \bigg)\bigg]^{N-1} dX dV.
\end{split}
\label{eq:part_time_max_dist}
\end{equation}
Note that the integrand in Eq. \eqref{eq:part_time_max_dist} is dominated by $f_x(X)f_v(V)$. Thus, by dominated convergence, one can determine the asymptotic behavior $P(t_i^{(N)} < \mu_N)$ by studying the large $N$ limit of Eq. \eqref{eq:cond_prob}.  We search for a scaling of $\mu_N$ depending on $N$ such that $$\bigg[P(t_i^{(N)} < \mu_N|(x_i,v_i)=(X,V))\bigg]^{N-1}=\bigg[P\bigg(\frac{x_1-X}{V-v_1} < \mu_N \bigg)\bigg]^{N-1}$$
has a non-trivial limit. We begin by constructing an integral representation of $\displaystyle P\bigg(\frac{x_1-X}{V-v_1} < \mu_N \bigg)$
by conditioning on the value of $v_1$.
\begin{equation}
\begin{split}
P\bigg(\frac{x_1-X}{V-v_1} < \mu _N\bigg) &=
 P\bigg(x_1>X+(V-v_1)\mu_N, V < v_i\bigg) + P\bigg(x_i<X+(V-v_1)\mu_N ,V\ge v_i\bigg)  \\
&=\int_V^\infty \int_{X+\mu_N(V-v_1)}^\infty f_x(x_1)f_v(v_1) dx_1 dv_1 
 + \int_{-\infty}^V \int_{-\infty}^{X+\mu_N(V-v_1)}f_v(v_1)f_x(x_1) dx_1 dv_1 \\
&=\int_V^\infty f_v(v_1)[1-F_x(X+\mu_N(V-v_1)]dv_1  +\int_{-\infty}^V f_v(v_1)F_X(X+\mu_N(V-v_1) )dv_1
\end{split}
\label{eq:cond_prob_int1}
\end{equation}
Applying the change of variables $w=v_1-V$ gives a more useful representation of Eq. \eqref{eq:cond_prob_int1}.
\begin{equation}
\begin{split}
P\bigg(\frac{x_1-X}{V-v_1} < \mu _N\bigg)
& =\int_0^\infty f_v(w+V) [1-F_x(X-\mu_N w)]dw +\int_{-\infty}^0 f_v(w+V) F_x(X-\mu_N w)dw \\
&=1 - \int_0^\infty f_v(w+V) F_x(X-\mu_N w ) dw  - \int_{-\infty}^0 f_v(w+V)[1-F_x(X-\mu_N w )]dw 
\end{split}
\label{eq:cond_prob_int2}
\end{equation}
Inputting the result from Eq. \eqref{eq:cond_prob_int2} in Eq. \eqref{eq:part_time_max_dist} gives an integral representation of the distribution of $t_i^{(N)}.$
\begin{equation}
\begin{split}
P(t_i^{(N)} < \mu_N) &= \int_{\mathbb{R}^2} f_x(X)f_v(V)\bigg[1-\int_0^\infty f_v(w+V) F_x(X-\mu_N w ) dw \\
& \hspace{4 cm}-\int_{-\infty}^0 f_v(w+V)[1-F_x(X-\mu_N w )]dw\bigg]^{N-1}dXdV 
\end{split}
\label{eq:tiN_dist_exact}
\end{equation}
The results of Thms. \ref{FTC_particle1} - \ref{FTC_particle3} depend on  the pointwise limit of 
\begin{equation}
g(X,V,N) = \bigg[1-\int_0^\infty f_v(w+V) F_x(X-\mu_N w ) dw -\int_{-\infty}^0 f_v(w+V)[1-F_x(X-\mu_N w )]dw\bigg]^{N-1} 
\label{eq:g(X,V,N)}
\end{equation}
as $N\to \infty$ under a suitable choice of $\mu_N.$ 
In particular, as $g(X,V,N)$ is a probability, hence bounded by one, we have the following proposition.
\begin{proposition}
Suppose that $\lim_{N\to\infty} g(X,V,N)$ exists. Then,
\begin{equation}
\begin{split}
\lim_{N\to \infty} P\bigg(t_i^{(N)}< \mu_N\bigg) &= \lim_{N\to \infty}\int_{\mathbb{R}^2} f_x(X)f_v(V) g(X,V,N) dX dV= \int_{\mathbb{R}^2} f_x(X) f_v(V) \bigg(\lim_{N\to \infty} g(X,V,N) \bigg) dX dV
\end{split}
\label{eq:FCT_particle_limit}
\end{equation}
\end{proposition}  

\noindent The details of the proofs of Thms.\ref{FTC_particle1} - \ref{FTC_particle3} consider the pointwise limit of $g(X,V,N)$ as $N\to \infty.$ \jld{Since the proof of Theorem 3 is similar to \jld{that of} Theorem 2, we have placed it in the appendix.}

\begin{proof}[Proof of Theorem \ref{FTC_particle1}] 
Suppose that $E|X_1| < \infty$ and $f_v$ is continuous and bounded.  Let $\mu_N = \mu N$ with $\mu>0.$  Consider the pointwise limit of $g(X,V,N)$ as $N\to \infty.$ 
\begin{equation}
\begin{split}
g(X,V,N) &= \bigg[1-\int_0^\infty f_v(w+V) F_x(X-\mu N w ) dw -\int_{-\infty}^0 f_v(w+V)[1-F_x(X-\mu N w )]dw\bigg]^{N-1} \\
&=\bigg[1-\frac{1}{ N}\int_0^\infty f_v(w/N+V) F_x(X-\mu w ) dw - \frac{1}{N} \int_{-\infty}^0 f_v(w/N+V)[1-F_x(X-\mu w )]dw \bigg]^{N-1}\\
&=\exp\bigg\{(N-1) \ln \bigg(1-\frac{1}{ N}\int_0^\infty f_v(w/N+V) F_x(X-\mu w ) dw \\
& \hspace{5 cm}- \frac{1}{N} \int_{-\infty}^0 f_v(w/N+V)[1-F_x(X-\mu w )]dw  \bigg)    \bigg\}
\end{split}
\end{equation}
Recall that $E|X|<\infty$, so both $\displaystyle \int_0^\infty F_x(X-\mu w)  dw$ and $\displaystyle \int_{-\infty}^0 [1-F_x(X-\mu w )] dw$
are finite. Additionally, 
\begin{equation*}
\begin{split}
0 &<  f_v(w/N +V) F_x(X-\mu w) \le M_v F_x(X-\mu w)  \\
0&< f_v(w/N +V)[1-F_x(X-\mu w )]  < M_v [1-F_x(X-\mu w )] \\
\end{split} 
\end{equation*}
where $M_v = \sup |f_v| $. By dominated convergence, it follows that 
\begin{equation*}
\begin{split}
\lim_{N\to \infty} \int_0^\infty f_v(w/N+V) F_x(X-\mu w ) dw 
&= \int_\mathbb{R} f_v(V) F_x(X-\mu w) dw \\
\lim_{N\to \infty} \int_{-\infty}^0 f_v(w/N +V)[1-F_x(X-\mu w )] dw
&= \int_{-\infty}^0 f_v(V)[1-F_x(X-\mu w )] dw
\end{split}
\end{equation*}
Hence, by the relation $\ln (1-x)= -x +o(x)$ for small $x$, then as $N$ tends to \jld{$\infty$},
\begin{equation}
\begin{split}
g(X,V,N) &= \exp \bigg\{ -\frac{N-1}{N}\bigg( \int_0^\infty f_v(V)F_x(X-\mu w ) dw+ \int_{-\infty}^0 f_v(V)[1-F_x(X-\mu w )]  \bigg) +o(1) \bigg\} \\
& \to \exp\bigg\{ -f_v(V) \bigg(\int_0^\infty F_x(X-\mu w ) dw+ \int_{-\infty}^0 [1-F_x(X-\mu w )]  dw  \bigg)  \bigg\}.\\
\end{split}
\end{equation}
Furthermore,
\begin{equation}
\begin{split}
&\int_0^\infty F_x(X-\mu w ) dw+ \int_{-\infty}^0 [1-F_x(X-\mu w )]  dw =\frac{1}{\mu }\int_0^\infty F_x(X-y)dy + \frac{1}{\mu} \int_{-\infty}^0 [1-F_x(X-y)] dy \\
&=\frac{1}{\mu} \bigg(\int_0^\infty y f_x(X-y)dy  - \int_{-\infty}^0 y f_x(X-y)dy    \bigg)=\frac{1}{\mu} \bigg( \int_{-\infty}^X (X-y) f_x(y) dy -\int_X^\infty (X-y) f_x(y) dy\bigg) \\
&=\frac{1}{\mu} \int_\mathbb{R} |X-y| f_x(y) dy
\end{split}
\label{eq:E|X-y|}
\end{equation}
So that for $\mu >0$, $\mu_N = \mu \cdot N$, and with the representation of $P(\frac{t_i^{(N)}}{N}\le \mu)$ from Eq. \eqref{eq:FCT_particle_limit}, we have the desired result\jld{:}
$$\lim_{N\to \infty} g(X,V,N) = \exp\bigg\{ -\frac{1}{\mu} f_v(V) \int_\mathbb{R} |X-y| f_x(y) dy\bigg\}. $$
Thus, by Proposition 3,
$$\lim_{N\to\infty}P\bigg(\frac{t_i^{(N)}}{N} < \mu \bigg) = \int_{\mathbb{R}^2} f_x(X)f_v(V) \exp\bigg\{ -\frac{1}{\mu} f_v(V) \int_\mathbb{R} |X-y| f_x(y) dy\bigg\} dX dV. $$
\end{proof}
\begin{proof}[Proof of Theorem \ref{FTC_particle2}] 

As in the proof of Thm. \ref{FTC_particle1}, we proceed by investigating the large $N$ behavior of $g(X,V,N)$.  However, unlike Thm. \ref{FTC_particle1}, we assume that $\{X_i\}_{i=1}^N$ have density $f_x(x) = \frac{C_\alpha}{1+|x|^{1+\alpha}}$ for $0<\alpha< 1$, so that $E|X_i|= \infty$ and $$\int_0^\infty F_x(X- w)  dw =\int_{-\infty}^0 [1-F_x(X- w )] dw=\infty.$$
\seth{Now for $\mu>0$, we take} $\mu_N = \mu\cdot N^{1/\alpha}$ and begin with the form of $g(X,V,N)$ from Eq. \eqref{eq:g(X,V,N)}.
\begin{equation}
\begin{split}
g(X,V,N) &= 
 \bigg[ 1 - \int_0^{N^{\alpha-2}} f_v(w+V)F_x(X-\mu N^{1/\alpha}  w) dw -\int_{-{N^{\alpha-2}} }^0 f_v(w+V)[1-F_x(X-\mu N^{1/\alpha}  w)] dw   \\
& \hspace{0.6 cm} - \int_{{N^{\alpha-2}} }^\infty f_v(w+V)F_x(X-\mu N^{1/\alpha}  w) dw- \int_{-\infty}^{-{N^{\alpha-2}} } f_v(w+V)[1-F_x(X-\mu N^{1/\alpha}  w)]dw \bigg]^{N-1}
\end{split}
\label{eq:g(X,V,N)_split}
\end{equation}
\jld{We now} investigate the large $N$ behavior of the different integrals within Eq. \eqref{eq:g(X,V,N)_split}.  The following statements follow from the \jld{boundedness} of $f_v$. For fixed $(X,V)$,
\begin{equation*}
\begin{split}
\int_0^{{N^{\alpha-2}} } f_v(w+V)F_x(X-\mu N^{1/\alpha} w) dw &= O(1)N^{\alpha-2} = o\bigg(\frac{1}{N}\bigg)\\
\int_{-{N^{\alpha-2}} }^0 f_v(w+V)[1-F_x(X-\mu N^{1/\alpha} w)] dw &= O(1) {N^{\alpha-2}}  = o\bigg(\frac{1}{N}\bigg)
\end{split}
\end{equation*}
For the remaining integrals in $g(X,V,N)$, we need some bounds on $F_x(-z)$ and $1-F_x(z)$ for large $z$.  Note that,
\begin{equation}
F_x(-z)= \int_{-\infty}^{-z} \frac{C_\alpha}{1+|t|^{1+\alpha}} dt=C_\alpha \int_{-\infty}^{-z} \frac{1}{|t|^{1+\alpha}} \frac{1}{1+\frac{1}{|t|^{1+\alpha}}} dt \\
\end{equation}
Thus, \seth{by the symmetry of $f_x$} we have the following bounds of $F_x(-z)=1-F_x(z)$
\begin{equation}
\frac{1}{1+\frac{1}{|z|^{1+\alpha}}}\frac{C_\alpha}{\alpha} |z|^{-\alpha}=\frac{C_\alpha}{1+\frac{1}{|z|^{1+\alpha}}} \int_{-\infty}^{-z}\frac{1}{|t|^{1+\alpha}} dt \le F_x(-z)=1-F_x(z) \le C_\alpha \int_{-\infty}^{-z} \frac{1}{|t|^{1+\alpha}}dt=\frac{C_\alpha}{\alpha} |z|^{-\alpha}.
\label{eq:FX_bounds}
\end{equation}
The upper and lower bounds converge as $z\to\infty$ so we have the following scaling laws for large $z$,
\begin{equation}
F_x(-z) =1-F_x(z) = \frac{C_\alpha}{\alpha} |z|^{-\alpha} + O\bigg(\frac{1}{|z|^{1+2\alpha}}\bigg) \\
\label{eq:scalings}
\end{equation}
Note that $N^{\alpha -2}\cdot N^{1/\alpha} = N^{(1-\alpha)^2/\alpha}.$  Thus, we can take $N$ large, \seth{depending on $\mu$ and $X$}, so that $\mu N^{1/\alpha} w - X\ge \mu N^{(1-\alpha)^2/\alpha} - X >> 1$ for all $w\ge N^{\alpha-2}$ allowing us to make use of the bounds from Eq. \eqref{eq:FX_bounds} in the remaining integrals in $g(X,V,N).$ Then, for fixed $(X,V)$,
\begin{equation}
\int_{N^{\alpha-2}}^\infty f_v(w+V) F_x(X-\mu N^{1/\alpha} w) dw 
=\int_{N^{\alpha-2}}^\infty f_v(w+V) \bigg(\frac{C_\alpha}{\alpha} \frac{1}{|X-\mu N^{1/\alpha} w |^\alpha} + O\bigg(\frac{1}{|X-\mu N^{1/\alpha}w|^{1+2\alpha}} \bigg) \bigg)dw \\
\end{equation}
Recall that $f_v$ is bounded. Let $M_v=\sup |f_v|$. We observe that the correction in the preceding integral is $o(1/N)$ as follows.
\begin{equation}
\begin{split}
\int_{N^{\alpha-2}}^\infty\frac{f_v(w+V) }{|X-\mu N^{1/\alpha} w|^{1+2\alpha}} dw &\le \int_{N^{\alpha-2}}^\infty\frac{M_v}{|X-\mu N^{1/\alpha} w|^{1+2\alpha}} dw \\
& = \frac{M_v}{2\alpha \mu N^{1/\alpha}} \cdot \frac{1}{|\mu N^{(1-\alpha)^2/\alpha} - X|^{2\alpha} } = O\bigg( \frac{1}{N^{1/\alpha+2(1-\alpha)^2}}\bigg)
\end{split}
\end{equation}
The quantity $2(1-\alpha)^2\ge 0$. Therefore, $\alpha \in (0,1)$ implies $1/\alpha+ 2(1-\alpha)^2 > 1$. Hence,
\begin{equation}
\begin{split}
\int_{N^{\alpha-2}}^\infty f_v(w+V) F_x(X-\mu N^{1/\alpha} w) dw&= \frac{1}{N} \frac{C_\alpha}{\alpha \mu^\alpha } \int_{N^{\alpha-2}}^\infty f_v(w+V) \frac{1}{w^\alpha} \frac{1}{|X/(\mu w N^{1/\alpha}) -1|^\alpha} dw + o\bigg(\frac{1}{N}\bigg) \\
& = \frac{1}{N} \frac{C_\alpha}{\alpha \mu^\alpha} \int_0^\infty \frac{f_v(w+V)}{w^\alpha} dw + o\bigg(\frac{1}{N}\bigg)
\end{split}
\end{equation}

Finally, for the final integral in Eq. \eqref{eq:g(X,V,N)_split}.  We make the following change of variables.
$$ \int_{-\infty}^{-{N^{\alpha-2}} } f_v(w+V)[1-F_x(X-\mu N^{1/\alpha}  w)]dw = \int_{N^{\alpha-2}}^\infty f_v(V-w)[1-F_x(X+\mu N^{1/\alpha} w)] dw.$$
To which we apply the same arguments as above to show
$$\int_{-\infty}^{-{N^{\alpha-2}} } f_v(w+V)[1-F_x(X-\mu N^{1/\alpha}  w)]dw = \frac{C_\alpha}{\alpha \mu^\alpha N} \int_0^\infty\jld{ \frac{f_v(V-w)}{w^{\alpha}}} dw + o\bigg(\frac{1}{N}\bigg)$$
Thus,
\begin{equation}
\begin{split}
&\lim_{N\to\infty} g(X,V,N)  = \jld{\lim_{N\to\infty}} \bigg(1- \frac{C_\alpha}{\alpha \mu^{\alpha} N} \int_0^\infty \frac{f_v(V+w)+f_v(V-w)}{w^\alpha}dw +o(1/N)  \bigg)^{N-1} \\
& = \exp \bigg\{ -\frac{C_\alpha}{\alpha \mu^\alpha} \int_0^\infty  \frac{f_v(V+w)+f_v(V-w)}{w^\alpha}dw  \bigg\}= \exp\bigg\{- \frac{C_\alpha}{\alpha \mu^\alpha} \int_{-\infty}^\infty \frac{f_v(V+w)}{|w|^\alpha} dw\bigg\}
\end{split}
\end{equation}
which in turns gives the desired result.
$$\lim_{N\to\infty} P\bigg( \frac{t_i^{(N)}}{N^{1/\alpha}} < \mu \bigg) = 
\int_{\mathbb{R}^2} f_x(X) f_v(V) \exp\bigg( -\frac{C(V)}{\mu^\alpha}\bigg) dX dV  $$
where 
$$C(V) = \frac{C_\alpha}{\alpha}\int_\mathbb{R} f_v(V+w)\frac{1}{|w|^\alpha} dw. $$
\end{proof}

\subsection{Proofs of Theorems Regarding $T^{(N)}$}\label{sec:proofs_TN}
The key observation in the proofs of Theorems \ref{FTC_system}, \ref{FTC_system2}, and \ref{FTC_system3} and Propositions \ref{prop:tight} is that for any time $z_N(t)\in\mathbb{R}$, the final collision time $T^{(N)}$ does not exceed $z_N(t)$ if and only if the random variable $A_{N,z_N(t)} = \sum_{1\le i<j\le N} \mathbf{1}_{\tau_{i,j} > z_N(t)}$ does not exceed one.  The distribution of $A_{N,z_N(t)}$ can be approximated by a Poisson distribution \cite{lao, silverman,Barbour_Eagleson}. \jld{W}e \seth{state} Corollary 2.1 from \cite{lao}.  
\begin{theorem1}
Let $\xi_1,\dots, \xi_N$ be $\mathbb{S}$-valued random variables and $h:\mathbb{S}^k \to \mathbb{R}$ be a symmetric Borel function. Let
\begin{equation*}
\begin{split}
p_{N,z} &= P(h(\xi_1,\dots,\xi_k)>z), \\
\lambda_{N,z}& = \binom{N}{k} p_{N,z}
\end{split}
\end{equation*}
If for some sequence of transformations, $z_N:\mathbb{T} \to \mathbb{R}$ with $\mathbb{T} \subset \mathbb{R}$, the conditions
\begin{equation*}
\lim_{N\to \infty} \lambda_{N,z_N(t)} = \lambda_t > 0
\end{equation*}
and
\begin{equation*}
\lim_{N\to\infty} N^{2k-1} P\bigg(h(\xi_1,\dots,\jld{\xi_k}) > z_N(t), h(\xi_{1+k-r},\dots, \xi_{2k-r}) > z_N(t)    \bigg) = \jld{0}
\end{equation*}
hold for all $t \in \mathbb{T} $ and $r=1,2,\dots,k-1$, then
\begin{equation*}
\lim_{N\to\infty} P(H_N \le z_N(t)) = \exp(-\lambda_t)
\end{equation*}
for all $t \in \mathbb{T} $ where $H_N = \max\limits_{1 \le j_1 \le \dots \le j_k\le N} h(\xi_{j_1},\dots, \xi_{j_k}).$
\end{theorem1}
\seth{Returning to our problem, let} $\xi_i = (X_i,V_i)$ be the initial position and velocity of particle $i$ and $h:\mathbb{S}^2 \to \mathbb{R}$ be the function which gives the path intersection time of two particles
$$h(\xi_i,\xi_j) = h\bigg((X_i,V_i),(X_j,V_j)\bigg) = \frac{X_i-X_j}{V_j-V_i}=\tau_{i,j}.$$
Then $H_N$ from the Silverman-Brown Limit Law stated above is the maximum collision time of the system, $T^{(N)}$.  To prove Theorems \ref{FTC_system}, \ref{FTC_system2}, and \ref{FTC_system3}, we need to find a time scale $z_N(t)$ such that 

\begin{align}
\binom{N}{2} P(\tau_{1,2} > z_N(t)) &\to \lambda_t >0 \label{eq:gen_req_a} \\
N^3P(\tau_{1,2} >z_N(t), \tau_{2,3} > z_N(t) ) &\to 0\label{eq:gen_req_b}.
\end{align}
\jld{The r}andom variable $A_{N,z_N(t)}$ \jld{previously introduced} is a summation of the $\binom{N}{2}$ dependent, Bernoulli random variables, $\mathbf{1}_{\tau_{i,j} > z_N(t)}$.  \seth{Recounting the idea of the proof of the Silverman-Brown law,} when Eq. \eqref{eq:gen_req_b} is satisfied, these random variable become sufficiently uncorrelated as $N\to\infty$ so that the distribution of $A_{N,z_N(t)}$ is approximately $Bin(\binom{N}{2},p_{N,z_N(t)})$.  The choice of $z_N(t)$ which gives a non-trivial limit in Eq. \eqref{eq:gen_req_a} allows one to approximate the distribution of $A_{N,z_N(t)}$ as $Poisson(\lambda_t)$. \seth{We mention} \jld{that} detailed error bounds for the Poisson approximation can be found in \cite{Barbour_Eagleson}.

\begin{proof}[Proof of Theorem \ref{FTC_system}]
Let $z_N(t) = \binom{N}{2} t$ and consider the behavior of $\lambda_{N,z_N(t)}$ as $N\to \infty.$
\begin{equation}
\begin{split}
\lambda_{N,z_N(t)} & = \binom{N}{2} P(\tau_{1,2} > z_N(t) ) =\binom{N}{2} \bigg(  1 - P(\tau_{1,2} \le z_N(t) ) \bigg)\\
&=\binom{N}{2} \bigg(\int_{\mathbb{R}^2}\int_0^\infty  f_x(X)f_v(V) f_v(w+V)F_x(X-z_N(t) w) dwdXdV \\
& \hspace{2.5 cm} +\int_{\mathbb{R}^2}\int_{-\infty}^0f_x(X)f_v(V) f_v(w+V) (1-F_x(X-z_N(t) w))dw dX dV  \bigg)
\label{Eq:ptau>z}
\end{split}
\end{equation}
where the last equality follows from Eq. \eqref{eq:tiN_dist_exact} taking $N=2$ and noting that $t_1^{(2)}=\tau_{1,2}$ \jld{(} recall that $\jld{t_1^{(2)}}$ is the (first and) last collision time in a system of two particles\jld{)}.  Making use of the change of variable $y=\binom{N}{2} w$ we have the following representation for $\lambda_{N,z_N(t)}$
\begin{equation}
\begin{split}
\lambda_{N,z_N(t)}  &= \int_{\mathbb{R}^2}\int_0^\infty f_x(X)f_v(V) f_v\bigg(V+\frac{y}{\binom{N}{2}}\bigg) F_x(X-t y)dy dX dV\\
& \hspace{2 cm}+ \int_{\mathbb{R}^2}\int_{-\infty}^0 f_x(X)f_v(V) f_v\bigg(V+\frac{y}{\binom{N}{2}}\bigg) (1-F_x(X-t y)) dy dX dV
\end{split}
\end{equation}
Under the assumption that $f_v$ is bounded and $E|X_1|<\infty$, we can evaluate $\lim_{N\to \infty} \lambda_{N,z_N(t)}$ by use of the dominated convergence theorem.
\begin{equation}
\begin{split}
\lim_{N\to \infty} \lambda_{N,z_N(t)} &= \int_{\mathbb{R}^2}\int_0^\infty f_x(X)(f_v(V))^2 F_x(X-ty) dy dX dV \\
& \hspace{2 cm} + \int_{\mathbb{R}^2}\int_{-\infty}^0 f_x(X) (f_v(V))^2 (1-F_x(X-ty)) dy dX dV \\
&=\bigg(\int_{-\infty}^\infty f_v^2(V) dV \bigg) \int_{-\infty}^\infty f_x(X) \bigg(\int_0^\infty F_x(X-ty) dy + \int_{-\infty}^0 (1-F_x(X-ty))dy   \bigg)dX \\
&=\frac{1}{t} \int_{-\infty}^\infty f_v^2(V) dV \cdot \int_{\mathbb{R}^2} |X-y|f_x(X)f_x(y) dy dX
\end{split} 
\end{equation}
where the last equality follows from the same argument as in Eq. \eqref{eq:E|X-y|}.  We have now determined the appropriate scaling to use in the Silverman-Brown limit law.  All \jld{that} remains to be shown is that $$N^3 P(\tau_{1,2} >z_N(t), \tau_{2,3}>z_N(t)) \to 0$$ as $N\to \infty.$ To prove this limit we first note the event $\{\tau_{1,2} >z_N(t), \tau_{2,3}>z_N(t)\}$ is equivalent to the event $\{t_2^{(3)} > z_N(t)\}$.  We can construct an integral representation for $P(t_2^{(3)} > z_N(t))$ following a similar argument used in the construction of Eq. \eqref{eq:tiN_dist_exact}.

\begin{equation}
\begin{split}
& N^3P(\tau_{1,2} >z_N(t), \tau_{2,3}>z_N(t)) \\
 =\ & N^3 \int_{\mathbb{R}^2}f_x(X)f_v(V)P\bigg(\frac{x_1-X}{V-v_1} > z_N(t), \frac{x_3-X}{V-v_3} \jld{> z_N(t)} \bigg| (x_2,v_2)=(X,V)\bigg) dXdV \\
 =\ & N^3 \int_{\mathbb{R}^2}f_x(X)f_v(V)\bigg[P\bigg(\frac{x_1-X}{V-v_1} > z_N(t) \bigg) \bigg]^2dXdV 
 \end{split}
 \label{eq:joint_int}
 \end{equation}
 From Eq. \eqref{eq:cond_prob_int2}, it follows that
 \begin{align*}
 P\bigg(\frac{x_1-X}{V-v_1} > z_N(t) \bigg) &= 1- P\bigg(\frac{x_1-X}{V-v_1} \le z_N(t) \bigg) \\
 &=\int_0^\infty f_v(w+V)F_x(X-z_N(t) w) dw+\int_{-\infty}^0 f_v(w+V) (1-F_x(X-z_N(t) w))dw 
 \end{align*}
 Substituting the above equality into Eq. \eqref{eq:joint_int}, we continue the calculation.
 \begin{equation}
 \begin{split}
 &N^3P(\tau_{1,2} >z_N(t), \tau_{2,3}>z_N(t)) \\
 &=N^3\int_{\mathbb{R}^2} f_x(X)f_v(V) \bigg( \int_0^\infty f_v(w+V)F_x(X-z_N(t) w) dw +\int_{-\infty}^0 f_v(w+V) (1-F_x(X-z_N(t) w))dw  \bigg)^2 dX dV \\
&=N^3\int_{\mathbb{R}^2} f_x(X)f_v(V) \bigg( \int_0^\infty f_v(w+V) P(X'\le X-z_N(t)w) dw + \int_0^\infty f_v(V-w)P(X'> X + z_N(t)w) dw \bigg)^2dX dV
\end{split}
\end{equation}
where $X'$ has the same distribution as $X_1$.  We now make use of the assumption that $E|X_1|^{3/2} <\infty$ . Also, recall $z_N(t) = \binom{N}{2} t.$ 
\begin{equation}
\begin{split}
&N^3P(\tau_{1,2} >z_N(t), \tau_{2,3}>z_N(t)) \\
&\le N^3\int_{\mathbb{R}^2} f_x(X)f_v(V) \bigg( \int_0^\infty f_v(w+V) P( |X-X'|\ge z_N(t) w) dw \\
&\hspace{4 cm}+ \int_0^\infty f_v(V-w)P(|X'- X| >  z_N(t)w) dw \bigg)^2dX dV \\
&=N^3\int_{\mathbb{R}^2} f_x(X)f_v(V) \bigg( \int_0^\infty [f_v(V-w)+f_v(V+w)] P(|X'-X| \ge z_N(t) w ) dw \bigg)^2 dXdV \\
&\le N^3 \int_{\mathbb{R}^2} f_x(X)f_v(V) \bigg(\int_0^\infty [f_v(V-w)+f_v(V+w)] \frac{ E_{X'} (|X'-X|^{3/4}\mathbbm{1}_{|X'-X|>z_N(t) w})}{(z_N(t) w)^{3/4}} dw \bigg)^2 dX dV \\
&= C_N(t)   \int_{\mathbb{R}^2} f_x(X)f_v(V) \bigg(\int_0^\infty \frac{f_v(V-w)+f_v(V+w)}{w^{3/4}} E_{X'} (|X'-X|^{3/4}\mathbbm{1}_{|X'-X|>z_N(t) w}) \jld{dw}\bigg)^2 dX  dV
\end{split}
\label{eq:Pnz_ineq}
\end{equation}
where \seth{$E_{X'}[\cdot]$ is expectation with respect to $X'$ and} $C_N(t) = \frac{N^3}{(N^2-N)^{3/2}} \frac{2^{3/2}}{t^{3/2}} \to \big(\frac{2}{t}\big)^{3/2}$ as $N\to\infty.$ Thus, to show $N^3P(\tau_{1,2}>z_N(t),\tau_{2,3}>z_N(t))\to 0$ as $N\to\infty$, we need to show the final integral in Eq. \eqref{eq:Pnz_ineq} approaches zero as $N\to\infty.$ The integrand can be dominated as follows.
\begin{equation}
\begin{split}
&f_x(X)f_v(V) \bigg(\int_0^\infty \frac{f_v(V-w)+f_v(V+w)}{w^{3/4}} E_{X'} (|X'-X|^{3/4}\mathbbm{1}_{|X'-X|>z_N(t) w}) \jld{dw} \bigg)^2  \\
&\le f_x(X)f_v(V) \bigg( \int_0^\infty \frac{f_v(V-w)+f_v(V+w)}{w^{3/4}}dw\bigg)^2 \bigg(E_{X'} (|X'-X|^{3/4}) \bigg)^2 \\
&\le f_x(X)f_v(V) \bigg( \int_1^\infty (f_v(V-w)+f_v(V+w)) \jld{dw} + \int_0^1 \frac{2M_v}{w^{3/4}} dw\bigg)^2 E_{X'} |X'-X|^{3/2} \\
&\le f_x(X)f_v(V)(1+8M_v)^2 E_{X'}(|X'|+|X|)^{3/2}
\end{split}
\label{eq:Pnz_dom}
\end{equation} 
where $M_v = \max\limits_{y} f_v(y).$   The final equation in Eq. \eqref{eq:Pnz_dom} is integrable since $X_1$ has a 3/2 moment.  Using dominated convergence twice,
\begin{equation}
\begin{split}
&\lim_{N\to \infty} \int_{\mathbb{R}^2} f_x(X)f_v(V) \bigg(\int_0^\infty \frac{f_v(V-w)+f_v(V+w)}{w^{3/4}} E_{X'} (|X'-X|^{3/4}\mathbbm{1}_{|X'-X|>z_N(t) w}) \jld{dw} \bigg)^2 dX  dV \\
&= \int_{\mathbb{R}^2} f_x(X)f_v(V) \bigg(\lim_{N\to\infty} \int_0^\infty \frac{f_v(V-w)+f_v(V+w)}{w^{3/4}} E_{X'} (|X'-X|^{3/4}\mathbbm{1}_{|X'-X|>z_N(t) w}) \jld{dw} \bigg)^2 dX  dV \\
&=\int_{\mathbb{R}^2} f_x(X)f_v(V) \bigg(\int_0^\infty \frac{f_v(V-w)+f_v(V+w)}{w^{3/4}} \lim_{N\to\infty} E_{X'} (|X'-X|^{3/4}\mathbbm{1}_{|X'-X|>z_N(t) w}) \jld{dw} \bigg)^2 dX  dV
\end{split}
\end{equation}
Since $E|X'|^{3/2}<\infty$, it follows that $$\lim_{N\to\infty} E\jld{\bigg(}|X'-X|^{3/4}\mathbbm{1}_{|X'-X|>z_N(t)w}\jld{\bigg) =} \ 0$$ for all $w> 0.$  Thus, we can conclude that $$N^3P\big(\tau_{1,2}>z_N(t), \tau_{2,3}>z_N(t) \big)\to 0$$ as $N\to\infty$. The final requirement of the Silverman-Brown limit law is satisfied, therefore

$$P\bigg(\frac{T^{(N)}}{\binom{N}{2}} \le \mu\bigg) \to  e^{-C/\mu }  $$
where $$C = \int_{\mathbb{R}^2} |x-y|f_x(x) f_x(y) dx dy \cdot \int_\mathbb{R} f_v^2(v) dv. $$

\end{proof}

\begin{proof}[Proof of Theorem \ref{FTC_system2}]
Applying the Silverman-Brown limit law used in the proof of Theorem \ref{FTC_system} again with $z_N(t) = N^{2/\alpha} t$, we need only verify the following two requirements:
\begin{align*}
&\binom{N}{2} P(\tau_{1,2} > N^{2/\alpha} t) \to \frac{C_\alpha}{2at^{\alpha}}\int_{\mathbb{R}^2} f_v(V) \frac{f_v(V+w)}{|w|^\alpha} dw dV \\
&N^3 P(\tau_{1,2} > N^{2/\alpha}t, \tau_{2,3} > N^{2/\alpha} t ) \to 0.
\end{align*}
We begin by making use of Eq. \eqref{Eq:ptau>z}
\begin{equation}
\begin{split}
\binom{N}{2} P(\tau_{1,2} > N^{2/\alpha} t ) &=\binom{N}{2} \bigg(\int_{\mathbb{R}^2}f_x(X)f_v(V) \int_0^\infty  f_v(w+V)F_x(X-N^{2/\alpha}t w) dwdXdV \\
& \hspace{2.5 cm} +\int_{\mathbb{R}^2}f_x(X)f_v(V) \int_{-\infty}^0 f_v(w+V) (1-F_x(X-N^{2/\alpha}t  w))dw dX dV  \bigg) 
\end{split}
\end{equation}
In the proof of Thm. \ref{FTC_particle2}, it was shown that for large $N$, 
\begin{equation}
\begin{split}
\int_0^\infty  f_v(w+V)F_x(X-N^{1/\alpha}t w) dw &= \frac{C_\alpha}{\alpha t^\alpha N } \int_0^\infty \frac{f_v(V+w)}{|w|^\alpha} dw  + o(N^{-1}) \\
 \int_{-\infty}^0 f_v(w+V) (1-F_x(X-N^{\jld{1}/\alpha}t  w))dw &= \frac{C_\alpha}{\alpha t^\alpha N } \int_0^\infty \frac{f_v(V-w)}{|w|^\alpha} dw  + o(N^{-1})
 \end{split}
 \label{eq:stablealpha_asym}
\end{equation}
Replacing $N$ with $N^2$, we have the following relationship for large $N$.
\begin{equation}
\begin{split}
\binom{N}{2} P(\tau_{1,2} > N^{2/\alpha} t ) & = \binom{N}{2} \int_{\mathbb{R}^2} f_x(X)f_v(V) \bigg(\frac{C_\alpha}{\alpha t^\alpha N^2} \int_\mathbb{R} \frac{f_v(V+w)}{|w|^\alpha}dw  +o(N^{-2}) \bigg) dXdV \\
&\to \frac{C_\alpha}{2\alpha t^\alpha} \int_{\mathbb{R}^2} f_v(V)\frac{f_v(V+w)}{|w|^\alpha} dw dV \jld{\qquad \text{as } N \to \infty.}
\end{split}
\end{equation}
For the second requirement, we make use of the same asymptotic \jld{expansion} as \jld{in} Eq. \eqref{eq:stablealpha_asym}.
\begin{equation}
\begin{split}
&N^3P(\tau_{1,2} >N^{2/\alpha} t , \tau_{2,3} > N^{2/\alpha} ) \\
&=N^3 \int_{\mathbb{R}^2}f_x(X)f_v(V) \bigg(\int_0^\infty  f_v(w+V)F_x(X-N^{2/\alpha}t w) dw +  \int_{-\infty}^0 f_v(w+V) (1-F_x(X-N^{2/\alpha}t  w))dw\bigg)^2 dX dV  \\
&= N^3 \int_{\mathbb{R}^2}f_x(X)f_v(V)  \bigg(\frac{C_\alpha}{\alpha t^\alpha N^2} \int_\mathbb{R} \frac{f_v(V+w)}{|w|^\alpha} dw +o(N^{-2}) \bigg)^2 dXdV \\
&= \frac{C_\alpha^2}{\alpha^2\jld{t^{2 \alpha}}N} \int_\mathbb{R}f_v(V) \bigg(\int_\mathbb{R} \frac{f_v(V+w)}{|w|^\alpha} dw\bigg)^2 dV +o(N^{-1}) \to 0 \jld{\qquad \text{as } N \to \infty.}
\end{split}
\end{equation}
Having satisfied the requirements of the limit law, the desired result follows and the proof is complete.
\end{proof}

\begin{proof}[Proof of Proposition \ref{prop:tight}] It is sufficient to show $$\lim_{A\to\infty} \limsup_{N\to\infty} P\bigg(\bigg|\frac{T^{(N)}}{\binom{N}{2}} \bigg| > A\bigg)=0.$$
Fix $A>0$, then
\begin{equation}
\begin{split}
P\bigg(\frac{T^{(N)}}{\binom{N}{2}} > A \bigg) &= P\bigg(\max_{i,j} \tau_{i,j}> \binom{N}{2} A \bigg) = P\bigg(\sum_{i,j} \mathbbm{1}_{\tau_{i,j} > \binom{N}{2} A} \ge 1 \bigg) \\
&\le E\bigg(\sum_{i,j} \mathbbm{1}_{\tau_{i,j} > \binom{N}{2} A}\bigg) =\binom{N}{2} P\bigg(\tau_{1,2} > \binom{N}{2} A\bigg) \\
&\xrightarrow{N\to\infty} \frac{1}{A} \int_{\mathbb{R}} f_v^2(v)dv \cdot \int_{\mathbb{R}^2} |x-y|f_x(x)f_x(y) dx dy 
\end{split}
\end{equation}
The convergence in the final line follows from arguments made in the proof of Thm. \ref{FTC_system}. Thus, $$\lim_{A\to \infty } \lim_{N\to\infty} P\bigg(\frac{T^{(N)}}{\binom{N}{2}} > A \bigg) =0.$$
Also, for any $A>0$,
\begin{equation}
\begin{split}
P\bigg(\frac{T^{(N)}}{\binom{N}{2}} < -A \bigg) &\le  P(T^{(N)} < 0) \\
&\le P(\tau_{1,2} < 0, \tau_{3,4} < 0, \dots, \tau_{2\left\lfloor N/2\right\rfloor -1, 2\left\lfloor N/2\right\rfloor}< 0) \\
&=\bigg[P(\tau_{1,2} < 0 )\bigg]^{\left\lfloor N/2\right\rfloor} \xrightarrow{N\to\infty} 0.
\end{split}
\end{equation}
Thus, $\displaystyle \lim_{A\to\infty} \limsup_{N\to\infty} P\bigg(\bigg|\frac{T^{(N)}}{\binom{N}{2}} \bigg| > A\bigg)=0.$
\end{proof}

\section*{Acknowledgments}
\ayd{A large number of simulations were required in the numerical reconstruction of the distribution of $T^{(N)}/\binom{N}{2}$ under elastic collisions and the investigation of $M_t$ and $M_T$ under non-elastic collisions. An allocation of computer time from the UA Research Computing High Performance Computing (HPC) at the University of Arizona is gratefully acknowledged.}
This material is based upon work supported ARO grant W911NF-14-1-0179. 

\appendix
\section{Proof of Theorem 3}


In this case, let $\mu_N =\mu  N \log (N)$.  We now proceed by splitting the integrals  $g(X,V,N)$, again  with the aim of extracting the $O(1/N)$ components of the integrals within $g(X,V,N)$ as $N\to \infty$. 
\begin{equation}
\begin{split}
g(X,V,N) & = \bigg[1 - \int_0^\infty f_v(V+w)F_x(X-\mu N \log N w)dw - \int_{-\infty}^0f_v(V+w)[1-F_x(X-\mu N \log N w)]dw \bigg]^{N-1}\\
&=\bigg[1-\int_0^\infty \bigg(f_v(V+w)F_x(X-\mu N \log N w) + f_v(V-w)[1-F_x(X+\mu N \log N w)]\bigg) dw \bigg]^{N-1}\\
&=\bigg[1- \int_0^{\frac{1}{N\sqrt{\log N}}} \bigg(f_v(V+w)F_x(X-\mu N \log N w) + f_v(V-w)[1-F_x(X+\mu N \log N w)]\bigg) dw  \\
&\hspace{.9 cm} \jld{-} \int_{\frac{1}{N\sqrt{\log N}}}^\infty \bigg(f_v(V+w)F_x(X-\mu N \log N w) + f_v(V-w)[1-F_x(X+\mu N \log N w)]\bigg) dw \bigg]^{N-1}
\end{split}
\label{eq:g(X,V,N)_split2}
\end{equation}

The first integral in Eq. \eqref{eq:g(X,V,N)_split2} over the interval $(0,\frac{1}{N\sqrt{\log N}})$ is $O(\frac{1}{N\sqrt{\log N}})$. In the remaining integral, we first make the change of variables $w\mapsto N\log N w.$ 
\begin{equation}
\begin{split}
&\int_{\frac{1}{N\sqrt{\log N}}}^\infty \bigg(f_v(V+w)F_x(X-\mu N \log N w) + f_v(V-w)[1-F_x(X+\mu N \log N w)]\bigg) dw  \\
&=\frac{1}{N\log N}\int_{\sqrt{\log N}}^\infty \bigg[  f_v\big(V+\frac{w}{N\log N}\big) F_x(X- \mu w)  + f_v\big(V-\frac{w}{N \log N}\big)[1-F_x(X+\mu w)]  \bigg] dw
\end{split}
\end{equation}
Since $w > \sqrt{\log N}$, we can make use of the scalings from Eq. \eqref{eq:scalings} for $F_x(X-\mu w)$ and $1-F_x(X+\mu w)$ \seth{valid for} $\alpha =1$ and continue the preceding calculation. \jld{[I removed factors of $\pi$ in the $O(1)$ terms below]} 
\begin{equation}
\begin{split}
&\frac{1}{N\log N}\int_{\sqrt{\log N}}^\infty \bigg[  f_v\big(V+\frac{w}{N\log N}\big) F_x(X- \mu w)  + f_v\big(V-\frac{w}{N \log N}\big)[1-F_x(X+\mu w)]  \bigg] dw\\
&=\frac{1}{N\log N}\int_{\sqrt{\log N}}^\infty \bigg[ \frac{ f_v\big(V+\frac{w}{N\log N}\big)}{\pi(\mu w -X)}\bigg(1+\frac{O(1)}{(\mu w -X)^2} \bigg)+ \frac{f_v\big(V-\frac{w}{N \log N}\big)}{\pi(\mu w +X)}\bigg(1+\frac{O(1)}{(\mu w+X)^2}\bigg)  \bigg] dw \\
&= \frac{1}{N\log N}\int_{\sqrt{\log N}}^{N\sqrt{\log N}} \bigg[ \frac{ f_v\big(V+\frac{w}{N\log N}\big)}{\pi(\mu w -X)}\bigg(1+\frac{O(1)}{(\mu w -X)^2} \bigg)+ \frac{f_v\big(V-\frac{w}{N \log N}\big)}{\pi(\mu w +X)}\bigg(1+\frac{O(1)}{(\mu w+X)^2}\bigg)  \bigg]dw \\
&\hspace{5 mm}+\frac{1}{N\log N}\int_{N\sqrt{\log N}}^\infty \bigg[ \frac{ f_v\big(V+\frac{w}{N\log N}\big)}{\pi(\mu w -X)}\bigg(1+\frac{O(1)}{(\mu w -X)^2} \bigg)+ \frac{f_v\big(V-\frac{w}{N \log N}\big)}{\pi(\mu w +X)}\bigg(1+\frac{O(1)}{(\mu w+X)^2}\bigg)  \bigg] dw
\end{split}
\end{equation}
We now bound the integral over the interval $(N\sqrt{\log N}, \infty)$.
\begin{equation}
\begin{split}
&\frac{1}{N\log N}\int_{N\sqrt{\log N}}^\infty  \frac{ f_v\big(V\pm\frac{w}{N\log N}\big)}{\pi(\mu w \mp X)}\bigg(1+\frac{O(1)}{(\mu w \mp X)^2} \bigg)dw \\
&\le \frac{1}{\pi(\mu N\sqrt{\log N} \mp X)}\bigg(1+\frac{O(1)}{(\mu N\sqrt{\log N}  \mp X)^2} \bigg)\cdot \frac{1}{N\log N}\int_{N\sqrt{\log N}}^\infty  f_v\big(V\pm\frac{w}{N\log N}\big)dw \\ 
&= \frac{1}{\pi(\mu N\sqrt{\log N} \mp X)}\bigg(1+\frac{O(1)}{(\mu N\sqrt{\log N}  \mp X)^2} \bigg) \bigg(\pm  F_v(\infty) \mp F_v\big(V\pm \frac{1}{\sqrt{\log N}}\big)\bigg) \\
&= O\bigg(\frac{1}{N\sqrt{\log N}}\bigg)
\end{split}
\end{equation}
To review, at this point we have shown that as $N$ tends to infinity
\begin{align*}
g(X,V,N) &= \bigg[1 - \frac{1}{N\log N}\int_{\sqrt{\log N}}^{N\sqrt{\log N}}  \frac{ f_v\big(V+\frac{w}{N\log N}\big)}{\pi(\mu w -X)}\bigg(1+\frac{C}{(\mu w -X)^2} \bigg)dw  \\
&\hspace{9mm} \jld{-} \frac{1}{N\log N}\int_{\sqrt{\log N}}^{N\sqrt{\log N}}\frac{f_v\big(V-\frac{w}{N \log N}\big)}{\pi(\mu w +X)}\bigg(1+\frac{C'}{(\mu w+X)^2}\bigg)  dw  + O\bigg(\frac{1}{N\sqrt{\log N}}\bigg)\bigg]^{N-1}
\end{align*}
The remaining integrals over the interval $(\sqrt{\log N},N\sqrt{\log N})$ account for $O(1/N)$ terms which give rise to the exponential, $e^{-\frac{2f_v(V)}{\pi\mu}}$ from the theorem. Note that the correction terms, $C/(\mu w \pm X)^2$, from the scaling in Eq. \eqref{eq:scalings} give rise to corrections of the order $C/\log N$ \seth{in the following display.}
\begin{equation}
\begin{split}
\int_{\sqrt{\log N}}^{N\sqrt{\log N}}  \frac{ f_v\big(V+\frac{w}{N\log N}\big)}{\pi(\mu w -X)}\bigg(1+\frac{O(1)}{(\mu w -X)^2} \bigg)dw  & = \bigg(1+\frac{O(1)}{\log N}\bigg)\int_{\sqrt{\log N}}^{N\sqrt{\log N}}  \frac{ f_v\big(V+\frac{w}{N\log N}\big)}{\pi(\mu w -X)}
\\
\int_{\sqrt{\log N}}^{N\sqrt{\log N}}\frac{f_v\big(V-\frac{w}{N \log N}\big)}{\pi(\mu w +X)}\bigg(1+\frac{O(1)}{(\mu w+X)^2}\bigg)  dw & =\bigg(1+\frac{O(1)}{\log N}\bigg) \int_{\sqrt{\log N}}^{N\sqrt{\log N}}\frac{f_v\big(V-\frac{w}{N \log N}\big)}{\pi(\mu w +X)} dw.
\end{split}
\label{EQ:g_remain}
\end{equation}
\jld{ The correction in Eq. \eqref{EQ:g_remain} therefore will be of order $O(1/\log N)$.} \seth{Now} focus on the integrals
$$\frac{1}{N\log N} \int_{\sqrt{\log N}}^{N\sqrt{\log N}}\frac{f_v\big(V\pm\frac{w}{N \log N}\big)}{\pi(\mu w \mp X)}dw$$
\seth{and deduce a limit.} \seth{We} apply the change of variable $w\mapsto \frac{1}{N\sqrt{\log N}} w$.
\begin{equation}
\frac{1}{N\log N} \int_{\sqrt{\log N}}^{N\sqrt{\log N}}\frac{f_v\big(V\pm\frac{w}{N \log N}\big)}{\pi(\mu w \mp X)}dw = \frac{1}{\sqrt{\log N}} \int_{\frac{1}{N}}^1 \frac{f_v\big(V\pm \frac{w}{\sqrt{\log N}}\big)}{\pi(\jld{\mu} N\sqrt{\log N} w \mp X)} dw 
\end{equation}
By the assumption that $f_v$ is continuous, it is uniformly continuous on the interval $[V,V+1]$. Thus $f_v\big(V\pm \frac{w}{\sqrt{\log N}}\big) = f_v(V) +o(1)$ as $N\to \infty.$ Then,
\begin{equation}
\begin{split}
\frac{1}{\sqrt{\log N}} \int_{\frac{1}{N}}^1 \frac{f_v\big(V\pm \frac{w}{\sqrt{\log N}}\big)}{\pi(\mu N\sqrt{\log N} w \mp X)} dw &= \frac{1}{\sqrt{\log N}} \int_{\frac{1}{N}}^1 \frac{f_v(V)+o(1)}{\pi (\jld{\mu} N\sqrt{\log N} w \mp X)} dw\\
&=\frac{f_v(V)+o(1)}{\pi \mu N \log N}  \log |\pi \mu N \sqrt{\log N} w \mp X|\bigg|_{w=1/N}^{w=1} \\
&=\frac{f_v(V)}{\pi \mu N} + o(1/N)
\end{split} 
\end{equation}
Thus, as $N\to\infty$
\begin{equation}
g(X,V,N)= \bigg[1-\frac{2f_v(V)}{\pi \mu N} + o\bigg(\frac{1}{N}\bigg)\bigg]^{N-1} \to \exp \bigg(-\frac{2f_v(V)}{\pi \mu} \bigg) \jld{\qquad \text{as } N \to \infty,}
\end{equation}
which concludes the proof.

\section{Proof of Theorem 6}
In this case, we take $z_N(t) = t N^2 \log N $.  We need only show the following two requirements of the limit law are satisfied:
\begin{align*}
&\binom{N}{2} P(\tau_{1,2} > tN^2 \log N  ) \to \frac{2}{\pi t} \int_\mathbb{R} f_v^2 (V) dV \\
&N^3 P(\tau_{1,2} > tN^2 \log N , \tau_{2,3} >t N^2 \log N  ) \to 0
\end{align*}
The previous statements follow naturally from the asymptotic results derived in the proof of Thm. \ref{FTC_particle3}.  Again, from Eq. \eqref{Eq:ptau>z} with $z_N(t) = t N^2 \log N.$
\begin{equation}
\begin{split}
\binom{N}{2} P(\tau_{1,2} > tN^2 \log N  )&=\binom{N}{2} \bigg(\int_{\mathbb{R}^2}f_x(X)f_v(V) \int_0^\infty  f_v(w+V)F_x(X-t wN^2 \log N) dwdXdV \\
& \hspace{1 cm} +\int_{\mathbb{R}^2}f_x(X)f_v(V) \int_{-\infty}^0 f_v(w+V) (1-F_x(X-t w N^2 \log N  ))dw dX dV  \bigg) 
\end{split}
\end{equation}
In the proof of Thm. \ref{FTC_particle3}, it was shown that for large $N$, 
\begin{equation}
\begin{split}
\int_0^\infty  f_v(w+V)F_x(X-tN\log N w) dw &= \frac{f_v(V)}{\pi t \jld{N}} + o(N^{-1}) \\
 \int_{-\infty}^0 f_v(w+V) (1-F_x(X-t N\log N w))dw &= \frac{f_v(V)}{\pi t N}+ o(N^{-1})
\label{eq:cauchy_asym}
\end{split}
\end{equation}
Replacing $N \log N$ with $\jld{2} N^2\log N$ \jld{and $t$ with $t/2$} in the previous statement, we have the following relationship for large $N$.
\begin{equation}
\begin{split}
\binom{N}{2} P(\tau_{1,2} > tN^2 \log N  ) & = \binom{N}{2} \int_{\mathbb{R}^2} f_x(X)f_v(V) \bigg(\frac{\jld{4} f_v(V)}{\pi t N^2} + o(N^{-2})\bigg) dX dV \\
& \to \frac{\jld{2}}{\pi t} \int_\mathbb{R} f_v^2(V) dv \jld{\qquad \text{as } N \to \infty.}
\end{split}
\end{equation}
For the second requirement, we make use of the same asymptotic \jld{expansion as in} Eq. \eqref{eq:cauchy_asym}.
\begin{equation}
\begin{split}
&N^3P(\tau_{1,2} >N^2\log N t , \tau_{2,3} >N^2\log N t ) \\
&=N^3 \int_{\mathbb{R}^2}f_x(X)f_v(V) \bigg(\int_0^\infty  f_v(w+V)F_x(X-tw N^2\log N ) dw \\
&\hspace{6 cm}+  \int_{-\infty}^0 f_v(w+V) (1-F_x(X-tw N^2\log N))dw\bigg)^2 dX dV  \\
&= N^3 \int_{\mathbb{R}^2}f_x(X)f_v(V)  \bigg(\frac{\jld{4}}{\pi t N^2} f_v(V) + o(N^{-2})\bigg)^2 dXdV \\
&= \frac{\jld{16}}{\pi\jld{^2} t\jld{^2} N} \int_\mathbb{R}f_v\jld{^3}(V)  dV +o(N^{-1}) \to 0  \jld{\qquad \text{as } N \to \infty,}
\end{split}
\end{equation}
which completes the proof.

\section{Counterexample to the proof of Thm. \ref{FTC_system} when $E|X|<\infty$, $E|X|^{3/2}=\infty$} 

In the proof of Thm. \ref{FTC_system}, \jld{we verified the} two requirements
\begin{align}
&\binom{N}{2}P\bigg(\tau_{1,2} > \binom{N}{2} t \bigg) \to \frac{1}{t} \bigg(\int_{\mathbb{R}}f_v\jld{^2}(V)  dV\bigg) \int_{\mathbb{R}^2} |X-y|f_x(X)f_x(y) dXdy \label{Eq:requirement1} \\
&N^3P\bigg(\tau_{1,2} > \jld{\binom{N}{2} t}, \tau_{2,3}>\binom{N}{2} t\bigg) \to 0
\label{Eq:requirement2}
\end{align}
as $N\to \infty$.  Recall that Eq. \eqref{Eq:requirement1} followed from the assumption that $E|X_1|<\infty$ and the continuity and boundedness of $f_v$.   However, Eq. \eqref{Eq:requirement2} required the additional assumption that $E|X_1|^{3/2}< \infty.$   

Suppose the initial positions of particles are distributed with a density $$f_x(x) = \frac{C_\alpha}{1+|x|^{1+\alpha}}$$ for $\alpha \in (1,3/2]$.  In this case, $E|X_1|<\infty$, but $E|X_1|^{3/2}$ does not exist.  We now proceed to show Eq. \eqref{Eq:requirement2} is not satisfied for this choice of density for the initial position.  From the proof of Thm. \ref{FTC_system}, recall
\begin{align*}
&N^3P\bigg(\tau_{1,2} > z_N(t), \tau_{2,3}>\binom{N}{2} t\bigg) \\
&=N^3\int_{\mathbb{R}^2} f_x(X)f_v(V) \bigg( \int_0^\infty f_v(w+V)F_x\bigg( X-\binom{N}{2}tw\bigg) dw \\
&\hspace{6 cm}+ \int_0^\infty f_v(V-w)\bigg[1-F_x\bigg(X + \binom{N}{2}tw\bigg)\bigg] dw \bigg)^2dX dV
\end{align*}
which is bounded below \jld{by}
\begin{align*}
&N^3\int_{\mathbb{R}^2} f_x(X)f_v(V)\bigg(\int_0^\infty f_v(V-w)\bigg[1-F_x\bigg(X + \binom{N}{2}tw\bigg)\bigg] dw   \bigg)^2dXdV \\
&=N^3\int_{\mathbb{R}} f_v(V)\int_0^\infty f_v(V-w_1)\int_0^\infty f_v(V-w_2) \\
&\hspace{3 cm}\times
\int_{\mathbb{R}}f_x(X)\bigg[1-F_x\bigg(X+\large{\binom{N}{2}}tw_1\bigg)\bigg]\bigg[1-F_x\bigg(X+\binom{N}{2}tw_2\bigg)\bigg]dXdw_1dw_2dV
\end{align*}
Furthermore, we can restrict the region of integration to $w_1,w_2\in (1,2)$ and $X< -N(N-1) t$.  In this case, $X+\binom{N}{2}tw\le 0$. Thus, by the symmetry of $f_x$, it follows that $1-F_x(X+\binom{N}{2}tw) \ge 1/2$. Then,
\begin{align*}
&N^3P\bigg(\tau_{1,2} > z_N(t), \tau_{2,3}>\binom{N}{2} t\bigg) \\
&\ge N^3\int_{\mathbb{R}}f_v(V)\int_0^\infty f_v(V-w_1)\int_0^\infty f_v(V-w_2) \\
&\hspace{3 cm}\times
\int_{\mathbb{R}}f_x(X)\bigg[1-F_x\bigg(X+\large{\binom{N}{2}}tw_1\bigg)\bigg]\bigg[1-F_x\bigg(X+\binom{N}{2}tw_2\bigg)\bigg]dXdw_1dw_2dV \\
&\ge N^3 \int_{-\infty}^{-N(N-1)t} f_x(X) \cdot \frac{1}{2}\cdot\frac{1}{2} dX \cdot \int_\mathbb{R}f_v(V)\int_1^2 f_v(V-w_1) dw_1 \int_1^2f_v(V-w_2) dw_2 dV \\
& = CN^3\int_{-\infty}^{-N(N-1)t}f_x(X)dX
\end{align*}
For $z \ll -1$, $\int_{-\infty}^{z}f_x(z) dz \approx \frac{1}{|z|^\alpha}.$  In this case, we have assumed $\alpha\in (1,3/2]$ so that as $N\to\infty$, 
$$CN^3 \int_{-\infty}^{-N(N-1)t}f_x(X)dX \to\frac{C}{t^\alpha} \frac{N^3}{N^{2\alpha}}> 0.$$
Thus, Eq. \eqref{Eq:requirement2} is not satisfied. As such, the proof of Thm. \ref{FTC_system} does not apply for this choice of density for the initial position of particles.  


\renewcommand*{\bibname}{\section*{References}}
\bibliographystyle{plain}

\end{document}